\documentclass[journal,10pt,twocolumn]{IEEEtran}
\usepackage{amsthm,amsmath,amssymb,mathtools,dsfont}
\usepackage{graphicx}
\usepackage{cite}
\usepackage{cleveref}
\usepackage{color}
\usepackage{enumerate}
\usepackage{support-caption}
\usepackage{url}
\usepackage{caption}
\usepackage{subcaption}
\usepackage{algorithm,algpseudocode}

\newtheorem{lemma}{Lemma}

\newtheorem{claim}{Claim}
\newtheorem{theorem}{Theorem}

\newtheorem{remark}{Remark}
\newtheorem{corollary}{Corollary}
\newtheorem{proposition}{Proposition}

\usepackage{booktabs,lipsum}
\usepackage{newtxmath}
\DeclareMathOperator{\Tr}{Tr}
\DeclareMathOperator{\Id}{Id}

\ifCLASSINFOpdf

\else

\fi

\begin{document}

\title{A Pursuit-Evasion Differential Game with Strategic Information Acquisition}


\author{Yunhan Huang$^{1}$  and Quanyan Zhu$^{1}$
\thanks{$^{1}$ Y. Huang and Q. Zhu are with the Department of Electrical and Computer Engineering,
        New York University, 370 Jay St., Brooklyn, NY.
        {\tt\small \{yh.huang, qz494\}@nyu.edu}}
}

\markboth{ArXiv Version}%
{Huang \MakeLowercase{\textit{et al.}}: A draft}
%



\maketitle

\begin{abstract}
This paper studies a two-person linear-quadratic-Gaussian pursuit-evasion differential game with costly but controlled information. One player can decide when to observe the other player's state. However, one observation of another player's state comes with two costs: the direct cost of observing and the implicit cost of exposing his state. We call games of this type a Pursuit-Evasion-Exposure-Concealment (PEEC) game. The PEEC game constitutes two types of strategies: The control strategies and the observation strategies. We fully characterize the Nash control strategies of the PEEC game using techniques such as completing squares and the calculus of variations. We show that the derivation of the Nash observation strategies and the Nash control strategies can be decoupled. We develop a set of necessary conditions that facilitate the numerical computation of the Nash observation strategies. We show, in theory, that players with less maneuverability prefer concealment to exposure. We also show that when the game's horizon goes to infinity, the Nash observation strategy is to observe periodically, and the expected distance between the pursuer and the evader goes to zero with a bounded second moment. We conducted a series of numerical experiments to study the proposed PEEC game. We illustrate the numerical results using both figures and animation. Numerical results show that the pursuer can maintain high-grade performance even when the number of observations is limited. We also show that an evader with low maneuverability can still escape if the evader increases his stealthiness.
\end{abstract}


%
\IEEEpeerreviewmaketitle

\section{Introduction}

Pursuit-Evasion (PE) refers to the problem in which one or more evaders try to escape from one or several pursuers. \textit{Berge}, in 1957, initiated a PE problem where evaders move in a prescribed trajectory and the pursuers track with certain control constraints. In 1965, Isaacs, recognized as the father of differential games, bridges the problem of PE and zero-sum differential game in his seminal work \cite{isaacs_differential_1965}. Propelled by early pioneers such as John Breakwell, Richard Bellman, Lev Pontryagin, and Yu-Chi Ho, the study of PE differential games, advancing in parallel with the theory of differential games and optimal control, have been flourished over the past half-century  \cite{ho1965differential,foley1974class,bagchi1981linear,bacsar1998dynamic,li2011defending,duncan2014linear,glizer2015linear,duncan2015linear,hayoun2016mixed,oyler2016pursuit,jagat2017nonlinear,talebi2017cooperative,pachter2019toward,lopez2019solutions,weintraub2020introduction,garcia2020two}. The motive behind PE games is not limited to physical entities pursuing one another. Various formulations of PE games empower the problem solving in other research areas such as robotics, sports/games, target defense and cybersecurity \cite{kehagias2014role,li2011defending,pachter2019toward,garcia2020two,huang2020dynamic,singh2101dynamic}.

Among the differential game studies, particular attention is paid to information patterns of the Linear-Quadratic-Gaussian (LQG) differential games. The information pattern of a dynamic game describes the available information to each player at each state for sequential decision. Two classical information pattern is the open-loop pattern, under which players only know the initial state of the game, and the feedback pattern with full information\cite{isaacs_differential_1965,ho1965differential}. As far as information patterns are concerned, there are essentially three possibilities: no information, perfect(exact) information, or partial information. These possibilities lead to nine different cases of two categories separated based on the symmetric of information for two-player games. Many efforts has been dedicated to tackling different cases of information pattern \cite{basar1976uniqueness,behn1968class,rhodes1969differential,bagchi1981linear,bernhard1988saddle,gupta2014common,clemens2019lqg}.

In studies of PE differential games, it is a common assumption that state information is available any time to both players \cite{ho1965differential,foley1974class,li2011defending,duncan2014linear,duncan2015linear,duncan2018linear,jagat2017nonlinear,talebi2017cooperative,pachter2019toward,lopez2019solutions,garcia2020two}.
However, in real-world applications, state information, especially information regarding one's opponent, is not always available and usually comes with a price to attain. Examples of situations where information are costly can be found in many scenarios. One is the price of sensing, which includes monetary expense such as power consumption, deployment costs, and etc. For example, a radar measurement can easily lead to megawatts of power usage. The recent booming shared economy also encourages decision makers to acquire information from third-party service providers who have pre-deployed sensors and charge a pay-as-you-go price as their sensing resources are used. Another is the cost of communication. The cost of communication can be prohibitive for long-distance remote decision making tasks such as spacecraft and satellite re-orbiting, control of unmanned combat aerial vehicles.

Apart from the monetary cost such as the price of sensing and the cost of communication, there are also indirect costs of observation. One such indirect cost is from stealth considerations. In military affairs, the innovation of more advanced and autonomous information and communication technologies has engendered a new revolution, making the battle in cyberspace as crucial as the ones in physical battlefields. The ability to remain stealthy and to be deceptive becomes the most valued characteristics of battlefield things. For example, submarines are equipped with active sonars and passive sonars to detect its surroundings. Active sonars can detect `quiet' objects that passive sonars are not able to detect. However, the use of active sonars may expose the submarine itself. Despite the increasing interaction of players in the information space, frameworks that can capture the intricacy of the information interactions between players are missing in the existing literature.

To fill the void, in this work, we study the controlled information structure of LQG PE differential games with a finite horizon, where players can decide at each stage whether to attain information or not, which we call a Pursuit-Evasion-Exposure-Concealment (PEEC) game. Acquiring information is referred to as ``making an observation'' here, which sometimes is called ``taking a measurement'' in some references\cite{bernhard1988saddle,gupta2014common,Maity2017ACC,maity2017linear,clemens2019lqg,huang2020infinite}. Each observation comes with a cost that whoever makes this observation has to pay. Besides the quantitative price, the player who chooses to observe the other player may also expose his state information. In real-world applications, the cost of observation may come from sensing and/or the cost of communication and stealth considerations. For example, a radar measurement can easily lead to megawatts of power usage and the measurer's exposure to the target. In the PEEC game, each player has to decide when to observe by developing the observation strategy and how to control by designing the control strategy.

One related area of research is PE games with limited sensing capabilities, where players have limited sensing capabilities that allow it to observe the other players only if they fall within its sensing range \cite{lin2015nash,bopardikar2008discrete,cruz2001game,galati2007effectiveness}. Different from these works, we focus on controlled sensing where players have control over when they need to sense and when they should not. Their sensing can be limited or prohibited due to high monetary cost or stealth considerations at certain time. 

The problem of controlled observations with costs has been studied in the context of finite-horizon optimal control \cite{cooper1971optimal}, infinite-horizon optimal control\cite{huang2020infinite}, and Markov decision process \cite{huang2019continuous}. Jan Geert Olsder studied costly observations in a discrete-time dynamic game setting \cite{olsder1977observation}, where each player at each step makes independent observation choices and obtains their private observations. The author proposes a matrix game to solve for a Nash observation strategy, whose derivation becomes prohibitive when the game's horizon increases. Hence, only a two-stage game problem is investigated. In \cite{huang2021cross}, the authors extended the framework of dynamic games with costly observations into the context of security problems in cyber-physical systems where one player chooses to observe, and the other chooses whether to jam the observation or not. Both \cite{olsder1977observation} and \cite{huang2021cross} focus on discrete-time dynamic games. Dipankar Maity et al. \cite{Maity2017ACC}, and \cite{maity2017linear} study dynamic games with controlled observations in a continuous-time setting, where each player can only choose to observe at a finite number of times. In \cite{Maity2017ACC} and \cite{maity2017linear}, each player receives their private observations, and one player's observation decision won't affect the other player's information set. Our work is different from \cite{olsder1977observation,huang2021cross,Maity2017ACC,maity2017linear} in three ways. First, we study the controlled observations in a PE differential game setting, where the two players have specific goals (one is chasing, the other is avoiding). And this situation can result in interesting interactions between the two players in terms of observation strategies. Second, our work deals with an information pattern that previous works have not investigated. That is when one player chooses to observe, his/her information also exposes to his opponent. Third, we fully characterize the Nash control strategies and develop a set of necessary conditions, with which we design a numerical algorithm to compute the Nash observation strategies.

The contributions of this work is summarized as follows.
\begin{enumerate}
    \item We propose a new type of PE differential game called the PEEC game, where both the pursuer and the evader don't know each other's state information and can decide when to observe it by paying a cost. This framework introduces the concept of controlled information to PE differential games, which expand the interactions between the pursuer and the evader to not just the physical layer but also the battlefield of information.
    
    \item We first leverage It\^{o}'s formula and completion of squares to obtain the Nash control strategy structure. We show that the Nash control strategies are the same as would be obtained in a perfect feedback setting. Next, we fully characterize the Nash control strategies for any given observation strategies using the calculus of variations. We show that the derivation of the Nash observation strategies and the Nash control strategies can be decoupled. The Nash control strategies have the certainty equivalence property and satisfy the separation principle. And the observation strategies are determined only by system characteristics. We show that players with less maneuverability prefer concealment to exposure and the optimal number of observations within a finite horizon is inversely proportional to the cost of observation. We develop a set of necessary conditions that helps characterize the Nash observation strategies, with which we design an effective numerical algorithm to compute the optimal observation strategy.
    
    \item We analyze the asymptotic propertis of the game under the ergodic cost criterion. We show that when the horizon of the game goes to infinity, it is optimal to observe/expose periodically. We characterize the optimal inter-sampling period and show that the distance between the pursuer and the evader is stabilizable in the mean with bounded second moment . 
    
    \item Leveraging the theoretical results, we numerically characterize the observation strategies. In numerical studies, we illustrate the pursuer and the evader's actions in the PEEC game using both figures and animation. The results show that a pursuer with higher maneuverability than the evader prefers more exposures/observations. But the pursuer can achieve reasonably good performance even when the number of observations is limited. The Nash observation strategy enables the pursuer to observe efficiently (observe less often while maintaining a good performance). We also show that when only a limited number of observations are available, larger system disturbances give an evader with less maneuverability more advantage. A less maneuverable evader can still escape if he/she can avoid being detected by his/her opponent frequently by making it more expensive for his/her opponent to observe.  
\end{enumerate}

\subsection{Notation}
In this paper, $\mathbb{R}$ represents the set of real numbers, $\mathbb{N}$ refers to the set of natural numbers including zero. Given any vector or matrix $x$, $x'$ means the transpose of $x$. Given any square matrix $M$, $\Tr(M)$ denote the trace of $M$. Given any vector $x$ and positive semi-definite matrix $Q$ with proper dimension, $\Vert x \Vert_Q = x'Qx$. Note that depending on the positive definiteness of $Q$, $\Vert \cdot\Vert_Q$ is not necessarily a norm. Let $M$ be any vector or matrix, $\dot{M}$ is the derivative of $M$ with respect to time. Given any two square matrix $M_1$ and $M_2$, $M_1\geq M_2$ meas $M_1 - M_2$ is positive semi-definite. Let $n$ be a positive integer, $\Id_n$ is a identity matrix with dimension $n\times n$.

\section{Formulation}\label{Sec:Formulation}

We consider a class of pursuit-evasion (PE) games described by the following linear stochastic differential equation:
\begin{equation}\label{StochasticDynamics}
    dx(t) = Ax(t)dt + B_p u_p(t)dt - B_e u_e(t)dt + C dw(t),
\end{equation}
with $x(0) = x_0$, where the initial position $x_0\in \mathbb{R}^n$ is not random and disclosed to both the pursuer and the evader, $x(t)\in \mathbb{R}^n$ captures the states (locations) of both players at time $t$. The terms $u_p(t)\in\mathbb{R}^{m_p}$ and $u_e\in \mathbb{R}^{m_e}$ denote respectively the control actions of the pursuer and the evader at time $t$. Here, $w(t)$ is a $q$-dimensional real-valued standard Wiener process independent of $x_0$. The positive integers $(n,m_p,m_e,q)$ are arbitrary. Moreover, $A$, $B_p$, $B_e$ and $C$ are real-valued matrices with appropriate dimensions. Let $\mathcal{I}_p(t)$ and $\mathcal{I}_e(t)$ be respectively the information available to the pursuer and the evader at time instance $t$. The family of admissible strategies for $u_p$ is $\mathcal{U}_p$, where $\mathcal{U}_p$ is a set of all possible $u_p$ such that $u_p$ is progressively measurable with respective to $\mathcal{I}_p(t)$ and square-integrable on $[0,T]$ almost surely. We define $\mathcal{U}_e$ in a similar way.

To characterize the objective of each player in classic PE games, we introduce a quadratic functional of $x$, $u_p$, and $u_e$, over a finite time horizon $[0,T]$:
$$
\begin{aligned}
J^0(u_p,u_e) &= \mathbb{E}\Big[ x(T)'Q_Tx(T)+\\
&\int_{0}^T x(t)' Q x(t) + u_p(t)' R_p u_p(t) - u_e(t)' R_eu_e(t)  \Big],
\end{aligned}
$$
where expectation $\mathbb{E}\left[\cdot\right]$ is over the statistics of $\{w(t),t\geq 0\}$; further, $Q$ and $Q_T$ are real-valued non-negative definite matrices, and $R_{p}$ and $R_{e}$ are real-valued positive definite matrices with appropriate dimensions. The objective of the pursuer is to find a $u_p \in \mathcal{U}_p$ that minimizes $J^0$ and the evader aims to do the opposite. In classic PE games, a common assumption is that the state history is fully observable to both players, i.e., $\mathcal{I}_p(t)=\mathcal{I}_e(t)=\{x(s),s\leq t\}$.

In this paper, we consider a PE game controlled information structure, in which both the purser and the evader can decide when to observe over the time interval $(0,T]$. Both players don't have the knowledge of the state unless they choose to observe. When a player decides to observe at time instance $t$, the player receives the state information $x(t)$. But the observation induces a non-negative cost and at the same time exposes the state information to the other player. The cost per observation is $O_p \in [0,\infty)$ for the pursuer and $O_e\in[0,\infty)$ for the evader. Let $\Omega_p = (N_p,\mathcal{T}_p)$ be the observation decisions of the pursuer, which include the number of observations made over time interval $(0,T]$, denoted by $N_p\in\mathbb{N}$ and the set of time instances when observations are made, denoted by $\mathcal{T}_p = \{t_{p,1}, t_{p,2},\cdots, t_{p,N_p}\}$. We have $\Omega_e = (N_e,\mathcal{T}_e)$ with $N_e\in\mathbb{N}$ and $\mathcal{T}_e = \{t_{e,1}, t_{e,2},\cdots, t_{e,N_e}\}$ defined similarly. The time instances when at least one of the players decides to observe is denoted by $\mathcal{T} = \mathcal{T}_p \cup \mathcal{T}_e$. Without loss of generality, we write $\mathcal{T}=\{t_1,t_2,\cdots,t_{N_p + N_e}\}$, where time instances in $\mathcal{T}$ are ordered as $0<t_1\leq t_2 \leq \cdots t_{N_p+ N_e}$. Since the observation made by one player will be exposed to the other player, the information available to the pursuer and the evader at time $t$ can be written as $\mathcal{I}(t)\coloneqq \mathcal{I}_p(t) = \mathcal{I}_e(t) =\{x(s)|0<s\leq t,s\in\mathcal{T}\}$. Therefore, the objective of the pursuer is to find an observation strategy $\Omega_p$ and an control strategy $u_p$ that minimize the following cost functional
\begin{equation}\label{CostFunctional}
\begin{aligned}
&J(\Omega_p,u_p, \Omega_e,u_e) = \mathbb{E}\Big[ O_p N_p - O_eN_e + x(T)'Q_Tx(T)+\\
&\ \ \ \ \int_{0}^T x(t)' Q x(t) + u_p(t)' R_p u_p(t) - u_e(t)' R_eu_e(t)  \Big].
\end{aligned}
\end{equation}
Meanwhile, the evader aims to minimize $J(\Omega_p,u_p, \Omega_e,u_e)$ with an optimal observation strategy $\Omega_e$ and an optimal control strategy $u_e$. The two players (the pursuer $\mathcal{P}$ and the evader $\mathcal{E}$), their strategies $(\Omega_p,u_p)$ and $(\Omega_e,u_e)$, the cost functional $J(\Omega_p,u_p, \Omega_e,u_e)$ in \cref{CostFunctional}, and the associated state dynamics given in \cref{StochasticDynamics} constitute a linear-quadratic-Gaussian zero-sum differential game with special controlled information structure, which we call a Pursuit-Evasion Exposure-Concealment (PEEC) game.

\begin{remark}
Our framework can capture pursuit-evasion differential games in various forms \cite{weintraub2020introduction,foley1974class,glizer2015linear,jagat2017nonlinear,li2011defending,bagchi1981linear,talebi2017cooperative,garcia2020two}. In pursuit-evasion differential games studied in these works , the pursuer and the evader usually have independent dynamics.  
$$
\begin{aligned}
dx_p(t) &= A_p x_p(t)dt + \tilde{B}_p u_p(t)dt + \tilde{C}_p dw_p(t),\\
dx_e(t) &= A_e x_e(t)dt + \tilde{B}_e u_e(t)dt + \tilde{C}_e dw_e(t),\\
\end{aligned}
$$
where $x_p(t)\in\mathbb{R}^{m}$ and $x_e(t)\in\mathbb{R}^{m}$. This general dynamics can be captured by our framework by defining $x = [x_p'\ x_e']'$, $w = [w_p' w_e']'$,  
$$
\begin{aligned}
&Q = \begin{bmatrix}
\Id_{m} & -\Id_{m}\\
-\Id_{m} & \Id_{m}
\end{bmatrix},\ \ \ C = \begin{bmatrix}
\tilde{C}_p & 0\\ 0 & \tilde{C}_e\end{bmatrix},\\
&B_p = \begin{bmatrix}\tilde{B}_p\\0
\end{bmatrix},\ \textrm{and}\ B_e = \begin{bmatrix}
0 \\ \tilde{B}_e
\end{bmatrix}.
\end{aligned}
$$
Note that this formulation yieds $x'Qx = \Vert x_p - x_e\Vert_2^2$, which describes the objectives of both the pursuer and the evader. This formulation has been adopted in \cite{foley1974class,li2011defending,garcia2020two,hayoun2016mixed}. Another way of formulating is letting $x = x_p -x_e$, when $A_p = A_e$. Let $Q = \Id_m$, we have $x'Qx = \Vert x_p -x_e \Vert_2^2$. This formulation has been used in \cite{glizer2015linear,jagat2017nonlinear,bagchi1981linear,talebi2017cooperative}.
\end{remark}

\begin{remark}
We consider a special information structure that is neither open-loop nor close-loop. The players have symmetric information. Both players have control over the information they receive, and one player's decision can affect the information the other player receives. The information the players have will further affect their control. The two players' observation strategies $(N_p,\mathcal{T}_p)$ and $(N_e,\mathcal{T}_e)$ decide the set of time instances $\mathcal{T}=\mathcal{T}_p\cup \mathcal{T}_e$ when information will be available. This set $\mathcal{T}$ determines $\mathcal{I}(t)$, which the controls have to be adapted to. Apart from $\mathcal{I}(t)$, it is tacitly assumed that the system characteristics 
$$\mathcal{I}_s = (A, B_p, B_e, C, x_0, Q_T, Q, R_p, R_e, O_p, O_e)$$ are known to both players.
\end{remark}


\section{Characterization of Nash Strategies}

In this section, we study the existence and the characterization of Nash strategies for the PEEC game. The Nash strategies involve the Nash observation strategies and the Nash control strategies selected by the two players. To characterize the Nash strategies, we first characterize the Nash control strategies for every possible observation strategies. That is for every possible $\Omega_p$ and $\Omega_e$, we characterize the Nash control strategies  $(u_p^*,u_e^*)\in \mathcal{U}_p\times \mathcal{U}_e$, such that
$$
J(\Omega_p,u^*_p, \Omega_e,u_e) \leq  J(\Omega_p,u^*_p, \Omega_e,u^*_e) \leq J(\Omega_p,u_p, \Omega_e,u^*_e),
$$
for every $u_p\in \mathcal{U}_p$ and $u_e\in \mathcal{U}_e$. Note that here the set of admissible control strategies $\mathcal{U}_p$ depends on the information structure $\mathcal{I}$, which is controlled by both players $\mathcal{P}$ and $\mathcal{E}$ through $(\Omega_p,\Omega_e)$. Hence, $u^*_p$ and $u^*_e$ also depends on $\Omega_p,\Omega_e$. Then, we write
\begin{equation}\label{DefineCostFunctionalUnderNash}
\tilde{J}(\Omega_p, \Omega_e) \coloneqq J(\Omega_p,u^*_p(\Omega_p,\Omega_e), \Omega_e,u^*_e(\Omega_p,\Omega_e)),
\end{equation}
where we emphasize the dependence of the Nash control strategies on $(\Omega_p,\Omega_e)$.

Next, we characterize the Nash observation strategies by finding a pair $(\Omega_p^*, \Omega_e^*)$ such that 
$$
\tilde{J}(\Omega^*_p, \Omega_e)  \leq \tilde{J}(\Omega^*_p, \Omega^*_e) \leq\tilde{J}(\Omega_p, \Omega^*_e),
$$
for all possible $\Omega_p$ and $\Omega_e$.

\subsection{The Nash control Strategies}

Suppose that we are given an arbitrary pair of observation strategies $(\Omega_p,\Omega_e)$. Due to the special information structure, instead of using dynamic programming techniques or Pontryagin’s type of approaches\cite{bacsar1998dynamic,engwerda2005lq}, we resort to a direct method to characterize the Nash control strategies. The direct method, widely applied recently in certain types of differential games\cite{duncan2014linear,duncan2015linear,duncan2018linear,maity2017linear}, is to form a generic structure of the cost functional in \cref{CostFunctional} by a standard completion of squares and characterize the Nash control strategies by using the calculus of variations type of techniques.

The following lemma is a result of applying It\^{o}'s lemma \cite{durrett2019probability} and a completion of squares on \cref{StochasticDynamics} and \cref{CostFunctional}. 
\begin{lemma}\label{Lemma:CostFunctionalSquare}
The cost functional $J$ in \cref{CostFunctional} associated with state dynamics \cref{StochasticDynamics} can be written as
\begin{equation}\label{CostFunctionalSquared}
\begin{aligned}
J &= \mathbb{E} \bigg[ \int_{0}^T \Vert  u_p(t)+ R_p^{-1}B_p'K(t)x(t)\Vert^2_{R_p} \\
&\ \ \ \ - \Vert  u_e(t)+R_e^{-1} B_e'K(t)x(t)\Vert^2_{R_e} dt\\
&\ \ \ \ + O_pN_p - O_eN_e \bigg] +\Vert x_0 \Vert^2_{K(0)} + \int_0^T  \Tr\left(K(t)CC'\right)dt,
\end{aligned}
\end{equation}
where $(K(t),t\in[0,T])$ is the symmetric non-negative solution of the Riccati equation:
\begin{equation}\label{RiccatiEquation}
-\dot{K}(t) = Q +  K(t)A + A' K(t) + K(t)\left( B_e R_e^{-1} B_e' - B_p R_p^{-1} B_p' \right) K(t),
\end{equation}
with $K(T) =Q_T$.
\end{lemma}

\begin{proof}
See \Cref{Appendix:CostFunctionalSquare}.
\end{proof}
To ensure the existence and the well-definedness of a solution $(K(t),t\in[0,T])$ defined by \cref{RiccatiEquation}, i.e., $K(\cdot)$ doesn't have a finite escape time in $[0,T)$, we assume that $B_p R_{p}^{-1} B_e \geq B_e R_{e}^{-1}B_e$ \cite{bacsar2008h}. The interpretation of this assumption in a PE game is that the pursuer has more maneuverability than the evader, otherwise the cost $J$ can go unbounded in finite time. 

In the classic PE game, the knowledge of the state $x(t)$ for all $t\in[0,T]$ is available to both players and there is no cost of observation, we can obtain a pair of Nash strategies $(u_p^*(t),u_e^*(t)) = (R_p^{-1}B_p'K(t)x(t),R_e^{-1}B_e'K(t)x(t))$, which yields a cost $\Vert x_0 \Vert_{K(0)} + \int_0^T \Tr(KCC')dt$.  However, in the PEEC game, the players have access to state information at only a finite number of time instances $\mathcal{T}$. Note that $\mathcal{T} = \mathcal{T}_p \cup \mathcal{T}_e$ depends on the observation strategies of both players. Recall that the observation strategies $\Omega_p,\Omega_e$ can be characterized by the number of observations $N_p,N_e$ and the time instances when an observation is made $\mathcal{T}_p,\mathcal{T}_e$.  The following theorem gives the Nash control strategies for every possible observation strategies $\Omega_p,\Omega_e$ of both players. The proof of the theorem follows the idea of forming a static game of infinite-dimensional action space and leveraging G\^{a}teaux derivative to check the first and second-order conditions of a Nash equilibrium (a saddle point in this zero-sum game.)

\begin{theorem}\label{Theorem:NashControlStrategy}
Given arbitrary $\Omega_p = (N_p,\mathcal{T}_p)$ and $\Omega_e=(N_e,\mathcal{T}_e)$. Let $\mathcal{T}=\mathcal{T}_p \cup \mathcal{T}_e=\{t_1,t_2,\cdots,t_{N_p+N_e}\}$ with $0<t_1\leq t_2 \leq\cdots,t_{N_p+N_e}<T$. Let $\mathcal{I}(t)= \{x(s)|0<s\leq t,s\in\mathcal{T}\}$ be the information available to $\mathcal{P}$ and $\mathcal{E}$ at time $t$. The Nash control strategies of the PEEC game defined by \cref{StochasticDynamics,CostFunctional} are
\begin{equation}\label{NashControlStrategies}
\begin{aligned}
&u_p^*(t) = - R^{-1}_p B'_p K(t)\hat{x}(t),\\
&u_e^*(t) = - R^{-1}_e B'_e K(t)\hat{x}(t),
\end{aligned}
\end{equation}
where $(K(t),t\in[0,T])$ is the solution of the Riccati equation \cref{RiccatiEquation} and  $(\hat{x}(t),t\in[0,T])$ is the solution of the following ordinary differential equation
\begin{equation}\label{EstimatorDynamic0}
\begin{aligned}
&d\hat{x}(t) = \left(A- (B_p R_p^{-1} B_p'- B_e R_e^{-1}B_e')K(t)\right)\hat{x}(t) dt,\\
&\hat{x}(0) = x_0,\ \hat{x}(\tau) = x(\tau),\ \textrm{for all }\tau \in \mathcal{T}.
\end{aligned}
\end{equation}
\end{theorem}

\begin{proof}
See \Cref{proof:NashControlStrategy}.
\end{proof}

\begin{remark}
If perfect feedback of state information is available, the Nash control strategies are the same as would be obtained in the absence of the additive disturbances. The missing feedback of state information is replaced by an estimate whose statistics is independent of the control. This separation principle also allows us to characterize Nash observation strategies separated from the control strategies.
\end{remark}

\begin{remark}
As we can see from \cref{EstimatorDynamic0}, between two neighboring observation time instances (say $t_i$ and $t_{i+1}$), two players are conducting open-loop control with initial condition $x(t_i)$. But the control is not open-loop for the entire horizon $[0,T]$. Whenever an observation is made, a close-loop information structure is formed at this particular time instance. The estimate then is reset to the actual state and the variance of the estimation error becomes zero. At extreme cases such as when $N_p = N_e = 0$, then $\mathcal{T} =\varnothing$, the Nash control strategies becomes an open-loop one. When $\mathcal{T} = [0,T]$, the Nash control strategies has close-loop information structure. In \Cref{Sec:NashObserStrategy}, we will discuss under what conditions these extreme cases are the Nash observation strategies.
\end{remark}

In the following Corollary, we substitute the Nash control strategies obtained in Theorem \ref{Theorem:NashControlStrategy} into the the cost functional \cref{CostFunctional}, which yields a cost functional that depends only on the pursuer and the evader's observation strategies.

\begin{corollary}\label{Corollary:RewriteCostFunctional}
Given arbitrary $\Omega_p = (N_p,\mathcal{T}_p)$ and $\Omega_e=(N_e,\mathcal{T}_e)$. Under the Nash control strategies $(\mu_p^*,u_e^*)$ given in \cref{NashControlStrategies} in \Cref{Theorem:NashControlStrategy}, the cost functional $\tilde{J}(\Omega_p,\Omega_e)$ defined in \cref{DefineCostFunctionalUnderNash} becomes
\begin{equation}\label{CostFunctionUnderNash}
\begin{aligned}
\tilde{J}(\Omega_p,\Omega_e) =&  \sum_{i=0}^{{N_p+N_e}} \int_{t_i}^{t_{i+1}} \Tr\left[\Sigma(t-t_i) \varphi(t)\right]dt + O_pN_p - O_eN_e\\ 
&+\Vert x_0 \Vert^2_{K(0)} + \int_0^T  \Tr\left(K(t)CC'\right)dt,
\end{aligned}
\end{equation}
where 
\begin{equation}\label{SigmaVarphi}
\begin{aligned}
\Sigma(t) &= \int_0^t  e^{A(t-s)} C C' e^{A(t-s)'}ds,\\
\varphi(t) &= K(t)(B_pR_p^{-1}B_p'-B_eR_e^{-1}B_e')K(t),
\end{aligned}
\end{equation}
and $0=t_0 <t_1 \leq t_2 \cdots \leq t_{N_p+N_e}<t_{N_p+N_e+1} =T$.
\end{corollary}

\begin{proof}
See \Cref{Appendix:RewriteCostFunctional}.
\end{proof}

Note that $t_1,t_2,\cdots,t_{N_p+N_e}\in\mathcal{T}$ are the ordered time instances at which at least one of the players choose to observe. Now we can see how the observation strategies of player $\mathcal{P}$ and player $\mathcal{E}$ affect the cost functional. The choices of observation points $\mathcal{T}_p$ and $\mathcal{T}_e$ gives $\mathcal{T} =\mathcal{T}_p\cup \mathcal{T}_e$, which is the set of time instances when state information will be available to both players and determines hence the information set $\mathcal{I}$. The control strategies, which are adapted to $\mathcal{I}$, will be affected. Since the last two terms in \cref{CostFunctionUnderNash} are constant, to study the Nash observation strategy, we only need to focus on the first three terms of \cref{CostFunctionUnderNash}.

\subsection{The Nash Observation Strategies}\label{Sec:NashObserStrategy}

In this section, we focus on characterizing the Nash observation strategies $(\Omega^*_p, \Omega^*_e)$. Following the results of \Cref{Corollary:RewriteCostFunctional}, the problem of characterizing a Nash observation strategy reduces to solving the following problem
\begin{equation}\label{CostFunctionalObservation}
\begin{aligned}
\min_{\Omega_p}\max_{\Omega_e} \tilde{J}_o(\Omega_p,\Omega_o) \coloneqq&  \sum_{i=0}^{{N_p+N_e}} \int_{t_i}^{t_{i+1}} \Tr\left[\Sigma(t-t_i) \varphi(t)\right]dt\\
&+ O_pN_p - O_eN_e,
\end{aligned}
\end{equation}
where $\left(\Sigma(t),t\in[0,T]\right)$ and $(\varphi(t),t\in[0,T])$ are defined in \cref{SigmaVarphi}. 

\begin{remark}
Here, $\Sigma(t-t_i)$ is the variance of the estimate error of the relative position between the pursuer and the evader at time $t$, where $t_i$ is the latest observation made before time $t$; $\varphi(t)$ can be interpreted as the matrix that scales the estimation error in different directions. The term $\Tr\left[\Sigma(t-t_i) \varphi(t)\right]$ captures the instantaneous cost at time $t$ induced by the mismatch between the actual relation position and the two players' estimates. The observation choices are control-aware by which we mean the estimation error $\Sigma(t)$ is scaled by the matrix $\varphi(t)$ and the matrix $\varphi(t)$ assign more weight to the estimation error corresponding to the states that are more information to control needs. From \cref{CostFunctionalObservation}, we know that the estimation error accumulates according to \cref{SigmaVarphi} until one of the player makes an observations. Once the observation is done, the estimation error is cleared. However, each observation made is subject to a cost $O_p$ or $O_e$ depending on who is the player that makes the observation. Hence, the pursuer and the evader have to make observation decision strategically over time. Overall, the observation decisions has to consider the trade-off between who suffers more from the estimation error (i.e., $\sum_{i=0}^{{N_p+N_e}} \int_{t_i}^{t_{i+1}} \Tr\left[\Sigma(t-t_i) \varphi(t)\right]dt$) and the costs of making observations (i.e., $O_pN_p - O_eN_e$).
\end{remark}

The observation strategies of player $\mathcal{P}$ involves $N_p$, the number of observations made in the time interval $[0,T]$, and $\mathcal{T}_p =\{t_{p,1},t_{p,2},\cdots,t_{p,N_p}\}$, the time instances when an observation is made. So does the observation strategies of player $\mathcal{E}$. The observation strategies of both players can be determined offline by solving the finite-dimensional minmax problem in \cref{CostFunctionalObservation}. The coupling between two player's observation strategies is introduced due to the fact that if one player choose to observe the other player's state, his/her own state information will be disclosed. To solve the problem in \cref{CostFunctionalObservation}, we first develop some structural results regarding the solution of the problem.

\begin{proposition}\label{Proposition:ObservationStrategies}
Consider the Concealment-Exposure(CE) game defined in \cref{CostFunctionalObservation}. Denote the Nash observation strategy of the CE game by $\left(\Omega_p^*=(N_p^*,\mathcal{T}_p^*),\Omega^*_e=(N_e^*,\mathcal{T}_e^*)\right)$ . If $B_p R_p^{-1} B_p' > B_e R_e^{-1} B_e'$, we have
\begin{enumerate}[(i)]
\item\label{Proposition:ZeroObservation} No matter what the observation strategy of the pursuer is, the best observation strategy for the evader $\mathcal{E}$ is to not observe, i.e., $N_e^*=0,\mathcal{T}^*_e=\varnothing$ for all $\Omega_p$.
\item\label{Proposition:ObservationNumUpperBonund} When $O_p =0$, it is optimal for the pursuer $\mathcal{P}$ to observe every time, i.e., $N_p^*=\infty,\mathcal{T}_p^*=[0,T]$. When $O_p>0$, the optimal number of observations for the pursuer $\mathcal{P}$ is upper bounded and inversely proportional to the observation cost $O_p$, i.e., 
\begin{equation}\label{ObservationNumberUpperBound}
N_p^* \leq \frac{1}{O_p} \int_0^T \Tr\left( \Sigma(t)\varphi(t)\right)dt.
\end{equation}
\item\label{ObservationTimeFirstOrder} The optimal observation time instances $\mathcal{T}_p^*$ for the pursuer $\mathcal{P}$ exist and need to satisfy
\begin{equation}\label{FirstOrderNecessaryCondition}
\begin{aligned}
&\int_{t^*_{p,i-1}}^{t^*_{p,i}}\Tr\left[ e^{A(t^*_{p,i}-t)}CC'{e^{A(t^*_{p,i}-t)}}'\varphi(t^*_{p,i})\right]dt\\
=& \int_{t^*_{p,i}}^{t^*_{p,i+1}} \Tr\left[e^{A(t-t^*_{p,i})} CC' {e^{A(t-t^*_{p,i})}}'\varphi(t) \right]dt,
\end{aligned}
\end{equation}
for $i=1,2,\cdots,N_p^*$.
\end{enumerate}
\end{proposition}

\begin{proof}
See \Cref{Appendix:ObservationStrategies}.
\end{proof}

\begin{remark}
In \Cref{Proposition:ObservationStrategies}, we focus on the case when $B_pR_p^{-1} B_p' - B_e R_e^{-1} B_e'>0$. When $B_pR_p^{-1} B_p' = B_e R_e^{-1} B_e'$, $\varphi(t)=0$ for all $t$. In this case, the CE game becomes $\min_{N_p}\max_{N_e} O_p N_p - O_eN_e$. The Nash observation strategies for both players are simply not to observe at all. When $B_pR^{-1}_pB_p<B_eR^{-1}_eB_e$, the solution of the Riccati equation in \cref{RiccatiEquation} admits a finite escape time \cite{bacsar2008h}. That means the PEEC game admits an unbounded value. Hence, discussing the observation strategies becomes meaningless in this case. Hence. in the remaining sections, we only focus on the case when $B_pR_p^{-1} B_p' - B_e R_e^{-1} B_e'>0$.
\end{remark}

From \Cref{Proposition:ObservationStrategies} (\ref{Proposition:ZeroObservation}), we know that when the pursuer has stronger maneuverability than the evader (i.e., $B_pR_p^{-1} B_p' - B_e R_e^{-1} B_e'>0$), the best observation strategy for evader is to stay stealthy, i.e., not observe, hence not expose him/herself. Results in  (\ref{Proposition:ObservationNumUpperBonund}) tell us that when there is no observation cost for the pursuer, i.e., $O_p=0$, since the pursuer has better maneuverability, the pursuer does not have any concerns about stealthiness. Hence, the pursuer will observe as often as possible. When the cost of observation is not zero, i.e., $O_p>0$, intimidated by the cost of sensing and communication, it is optimal that the pursuer observes only a finite number of times. The optimal number of observation times is inversely proportional to the observation cost $O_p$. When an arbitrary number of observation time instances $N_p$ is given, in (\ref{ObservationTimeFirstOrder}), we characterize the set of optimal observation time instances $\mathcal{T}^*_p$ using the first-order necessary conditions. And the set of time instances $\mathcal{T}_p^*$ that satisfies \cref{FirstOrderNecessaryCondition} is unique. From \cref{FirstOrderNecessaryCondition}, we can see that the optimal observation time instances are spread out over the horizon $[0,T]$. Given a limited number of observations over the horizon, it is unwise to allocate two observation instances in a short period of time. For each neighboring pair of observation instances $(t_{p,i-1}^*,t_{p,i}^*)$, the next neighboring pair of observation instances $(t_{p,i}^*,t_{p,i+1}^*)$ needs to be well separated such that the integral in the right side of \cref{FirstOrderNecessaryCondition} is equal to that of the left.

\subsection{Computational Methods}

In \Cref{Proposition:ObservationStrategies}, we show the existence of a Nash observation strategy and partially characterized a Nash observation strategy via theoretical analysis. More specifically, we characterize the evader's strategy, derive an upper bound on the optimal observation times  $N^*_p$ of the pursuer, and develop a set of necessary conditions for the optimal observation time instances $\mathcal{T}^*_{p}$. For a finite $T$, to fully characterize a Nash observation $\Omega^*_p$, we need to solve the following finite-dimensional optimization problem:
\begin{equation}\label{innerObservationOptimization}
\begin{aligned}
F_p(N_p) \coloneqq \min_{\substack{t_{p,i},\\1\leq i\leq N_p}} &\sum_{i=0}^{{N_p}} \int_{t_{p,i}}^{t_{p,i+1}} \Tr\left[\Sigma(t-t_{p,i}) \varphi(t)\right]dt + O_pN_p,\\
\textrm{     s.t.}\ \ \ \ \ \ \ \ \ \  &t_{p,0}=0,\ t_{p,N_p+1}=T,\\
&t_{p,i}
\leq t_{p,i+1},i=0,1,2,\cdots,N_p,\\
\end{aligned}
\end{equation}
where $F_p(N_p)$ is the optimal value of the CE game when the number of observations made is $N_p$. The first-order necessary conditions of this problem is provided in \cref{FirstOrderNecessaryCondition}. In general, a closed-form solution for the optimization problem in \cref{innerObservationOptimization} is unattainable. Since the first and second-order differentials of the objective function in \cref{innerObservationOptimization} can be expressed explicitly and the problem has only linear inequality constraint, we can leverage either first-order and second-order numerical optimization methods \cite{gill2021numerical} to find the optimal observation instances. 

However, the properties of \cref{FirstOrderNecessaryCondition} provide an alternative method to numerically compute the optimal observation instances $t^*_{p,1}, t_{p,2}^*,\cdots, t^*_{p,N_p}$. To more specific, \cref{FirstOrderNecessaryCondition} indicates that once $t_{p,1}^*$ is provided, $t_{p,2}^*$ can be computed easily. So can $t_{p,3}^*,\cdots, t_{N_a}^*$. Based on this feature, we propose a binary search algorithm that solves problem (\ref{innerObservationOptimization}) with a given $N_p$. In Algorithm \ref{algorithm1}, we aim to find a $t_{1}^\star$ such that $\vert t_{1}^\star - t_{p,1}^*\vert <\epsilon/2$. Line $1$ initializes all the parameters in (\ref{innerObservationOptimization}). Line $2$ sets the initial low bound $t_{low}$ and upper bound $t_{up}$ of $t_{p,1}^*$ to be $0$ and $T$ respectively. The initial guess of $t_1$ is $(0+T)/2$. Line $5$ computes the left-hand side of (\ref{FirstOrderNecessaryCondition}), which we rewrite as 
\begin{equation}\label{eq:FirstOrderNecessaryConditionLeft}
    l_p(t_{i-1},t_i)=\int_{t_{i-1}}^{t_{i}}\Tr\left[ e^{A(t_{i}-t)}CC'{e^{A(t_{i}-t)}}'\varphi(t_{i})\right]dt.
\end{equation}
Line $6$ computes the right-hand side integral in (\ref{FirstOrderNecessaryCondition}) from $t_i$ to $T$, which we write as 
\begin{equation}\label{eq:FirstOrderNecessaryConditionRight}
    r_p(t_i,T)=\int_{t_{i}}^{t_{f}} \Tr\left[e^{A(t-t_{i})} CC' {e^{A(t-t_{i})}}'\varphi(t) \right]dt.
\end{equation}
Line $7$-$11$ says for any $t_i,i=1,2,\cdots,N_a$ that is computed based on our guess $t_1$, if $ r_a(t_i,T)<l_a(t_{i-1},t_i)$, then our guess $t_1$ is larger than $t_{p,1}^*$. Hence, we set the upper bound $t_{up}$ as $t_1$ and reset out guess $t_1$ as $t_1 = (t_{low} + t_1)/2$. Then we break the for loop and start with our new guess $t_1$. Line $12$ computes the next observation instance using (\ref{FirstOrderNecessaryCondition}). Line $13$-$21$ says that when the for loop gets to $i=N_a$, we compute $t_{N_a+1}$. If $t_{N_a+1}<T$, our guess $t_1$ must be smaller than $t_1^*$. Hence, we set $t_{low} = t_1$, let our new guess to be $t_1 = (t_{up}+t_1)/2$, and breaks the for loop. If $t_{N_a+1} = T$ (it is impossible that $t_{N_a+1}>T$ due to our operations in Line $5$-$11$), then $t_1 = t_1^*$. Hence, we set $t_{low} = t_{up} =t_1$ to leave the while loop. Since the while ends when $\vert t_{up} - t_{low}\vert <\epsilon$,  we can ensure $\vert t_1^\star -t_{p,1}^* \vert<\epsilon/2$, where $t_{p,1}^*$ is the optimal first observation instance and $t_1^\star$ is the first observation instance found using Algorithm \ref{algorithm1}. The number of iterations needed for the while loop is less than $\min\{n\ |\ T/2^n\ \leq \epsilon\}$. For example, only $20$ iterations are needed to achieve $\epsilon = 10^{-5}$ when $T = 10$. Once $t_1^\star$ is obtained, the rest observation instances can be computed easily using  (\ref{FirstOrderNecessaryCondition}). Note that with $F_p^*(N_p)$ being computed for some small $N_p$, a bound similar to yet tighter than (\ref{ObservationNumberUpperBound}) can be developed. For example, when $F_p^*(N_p)$ is computed for $N_p=1,2,3$, if $N_p^*>3$, we have $F^*_p(3)+3O > O N_p^*$, i.e., $N_p^* -3 \leq F_p^*(3)/O$. Hence, we only need to compute $F_p^*(N_a)$ for a very limited number of $N_p$.

\begin{algorithm}[t]
\caption{Optimal Observation Instances Algorithm Based on Binary Search}\label{algorithm1}
\begin{algorithmic}[1]
\State Initialize $A$,$C$, $N_p$,$\varphi(\cdot)$, $T$, and tolerate$,\epsilon>0$
\State Set $t_{low}=0$, $t_0=0$, $t_{up}=T$, and $t_1 = (t_{up}+t_{low})/2$
\While {$|t_{up}-t_{low}|>\epsilon$}
\For{$i = 1,\cdots,N_p$}
\State  Compute $\textrm{val} = l_p(t_{i-1},t_i)$ defined in (\ref{eq:FirstOrderNecessaryConditionLeft})
\State Compute $\textrm{val}' = r_p(t_i,T)$ defined in (\ref{eq:FirstOrderNecessaryConditionRight})
\If{$\textrm{val}' < \textrm{val}$}
\State $t_{up} = t_{1}$
\State $t_1 = (t_{low}+t_1)/2$
\State \textbf{break}
\EndIf
\State Compute $t_{i+1}$ using (\ref{FirstOrderNecessaryCondition})
\If{$i=N_p$} 
\If{$t_{i+1}<T$}
\State $t_{low} =t_1$
\State $t_{1} = (t_{up} + t_1)/2$
\State \textbf{break}
\Else
\State $t_{low} = t_{up} = t_1$
\EndIf
\EndIf
\EndFor
\EndWhile
\State \textbf{return} $t_1^\star = (t_{low} + t_{up})/2$
\end{algorithmic}
\end{algorithm}

\begin{remark}
The discussion so far allows the pursuer to determine his/her observation strategy offline. To find an online implementation of the observation strategy, we can leverage dynamic programming techniques. We can first define 
$$
V(t) = \min_{N_p}  \min_{\substack{t_{p,i}\in [t,T],\\i=1,\cdots,N_p}} \sum_{i=0}^{{N_p}} \int_{t_{p,i}}^{t_{p,i+1}} \Tr\left[\Sigma(t-t_{p,i}) \varphi(t)\right]dt + O_pN_p,
$$
with $t_{p,0}=t$, $t_{p,N_p+1}=T$, and $V(T)=0$. Then, we need to show that 
$$
V(t) = \min_{\Delta_t\leq T-t} \left[ \int_{0}^{\Delta_t} \Tr\left( \Sigma(s)\varphi(t+s) \right) ds + O_p + V(t+\Delta_t)\right],
$$
where $V(\cdot)$ can be characterized by using techniques like approximate dynamic programming. With $V(\cdot)$ being characterized, whenever an observation is made, say an observation is made at time $t$, the pursuer can thus determine online the optimal waiting time for next observation $\Delta_t^*$ by solving
$$
\Delta_t^*=\arg\min_{\Delta_t\leq T-t} \left[ \int_{0}^{\Delta_t} \Tr\left( \Sigma(s)\varphi(t+s) \right) ds + O_p + V(t+\Delta_t)\right].
$$
The analysis of the dynamic programming approach and online implementation is out the scope of this paper. We leave it for future work.
\end{remark}

\subsection{Asymptotic Properties}

Let the terminal time $T$ go to infinity and consider the long-term average (e.g., ergodic) cost criterion of \cref{CostFunctional} with $Q_T=0$. The Nash control strategies for the ergodic criterion can be obtained by following similar steps of \Cref{Theorem:NashControlStrategy}. The Nash control strategies are stationary:
\begin{equation}\label{AsymptControl}
\begin{aligned}
&u_p^*(t) = - R^{-1}_p B'_p \tilde{K}\tilde{x}(t),\\
&u_e^*(t) = - R^{-1}_e B'_e \tilde{K}\tilde{x}(t),
\end{aligned}
\end{equation}
where $\tilde{K}\in\mathbb{R}^{n\times n}$ is the solution of the algebraic Riccati equation 
$$
Q +  \tilde{K}A + A' \tilde{K} - \tilde{K}\left(  B_p R_p^{-1} B_p' - B_e R_e^{-1} B_e'  \right) \tilde{K}=0.
$$
Between every two neighboring observation instances $[t_i,t_i+1)$, both players have open-loop estimate $\hat{x}(t)$ satisfying
\begin{equation}\label{AsymptoEstimate}
\begin{aligned}
&d\tilde{x}(t) = \left(A- (B_p R_p^{-1} B_p'- B_e R_e^{-1}B_e')\tilde{K}\right)\tilde{x}(t) dt,\\
&\tilde{x}(t_i) = x_{t_i},\ \ \ \textrm{for }t\in[t_i,t_{i+1}).
\end{aligned}
\end{equation}

We are also interested in the observation strategies under the long-term average cost criterion. From \cref{FirstOrderNecessaryCondition}, we can see that the optimal observation time instances distributed evenly over the time horizon when $K(t)$ becomes stationary at $\tilde{K}$, i.e., $t^*_{p,i+1}-t^*_{p,i} = t^*_{p,i}-t^*_{p,i-1}$ for $t_{p,i-1}^*,t_{p,i}^*,t^*_{p,i+1}$ for every $i$. Hence, when $T$ goes to infinity, the Nash observation strategy is for the pursuer to observe periodically. To find the Nash observation strategy, it is sufficient to find the optimal period $\Delta T$. Indeed, under periodic observations with inter-sampling duration $\Delta T$, the pursuer needs to solve the following optimization problem
\begin{equation}\label{OptimizationPeriod}
\min_{\Delta T > 0}\ \ \ \ f_p(\Delta T)\coloneqq \frac{1}{\Delta T} \int_{0}^{\Delta T} \Tr[\Sigma(t)\tilde{\varphi}] dt + \frac{1}{\Delta T} O_p.
\end{equation}
The first-order necessary condition gives that the optimal period $\Delta T^*$ satisfies
\begin{equation}\label{FirstOrderPeriod}
\Delta T^* \Tr[\Sigma(\Delta T^*)\tilde{\varphi}] - \int_0^{\Delta T^*}\Tr [\Sigma(t)\tilde{\varphi}]dt = O_p,
\end{equation}
which can be easily solved numerically.
Taking second derivative of the objective function $f_p$ with respect to $\Delta T$ yields
{\small
\begin{equation}\label{SecondOrderPeriod1}
\begin{aligned}
\frac{d^2}{d\Delta T^2} f_p(\Delta T) =& \frac{2}{\Delta T^3} \left(\int_0^{\Delta T}\Tr[\Sigma(t)\tilde{\varphi}]dt + O_p \right)\\
&- \frac{2}{\Delta T^2} \Tr[\Sigma(\Delta T)\tilde{\varphi}] + \frac{1}{\Delta T} \Tr[e^{A\Delta T}C C' e^{A\Delta T'}\tilde{\varphi}].
\end{aligned}
\end{equation}}
Substituting \cref{FirstOrderPeriod} into \cref{SecondOrderPeriod1} yields 
$$
\frac{d^2}{d\Delta T^2} f_p(\Delta T) = \frac{1}{\Delta T} \Tr[e^{A\Delta T}C C' e^{A\Delta T'}\tilde{\varphi}] >0.
$$
Also note that the left hand side of \cref{FirstOrderPeriod} is increasing in $\Delta T^*$. Hence, the optimal period $\Delta T^*$ that satisfies \cref{FirstOrderPeriod} is unique. Then we can conclude that in the infinite-horizon case with averaged cost, the optimal observation instances are $t_{p,i}^* =i\Delta T^*$ for $i=1,2,\cdots$.

{\bf Stability Properties:} Under the control strategies defined by \cref{AsymptControl} and \cref{AsymptoEstimate} and the periodic observation strategy $t_{p,i}^* =i\Delta T^*$ for $i=1,2,\cdots$, the pursuer can ensure the expected distance between the pursuer and the evader goes to $0$ with a bounded variance as time goes to infinity. That is $\mathbb{E}[x_t]\rightarrow 0$ as $t\rightarrow \infty$ and $\sup_{t\geq 0}\mathbb{E}[\Vert x_t \Vert^2] < \infty$.

From \cref{AsymptControl} and \cref{AsymptoEstimate}, the closed-loop system can be written as
$$
\begin{bmatrix}
dx(t) \\ d\hat{x}(t)
\end{bmatrix}= \bar{A}\begin{bmatrix}
dx(t) \\ d\hat{x}(t)
\end{bmatrix} + \begin{bmatrix}
Cdw(t)\\0
\end{bmatrix},
$$
for $t\in[i\Delta T^*,(i+1)\Delta T^*)$ and $x(i\Delta T^*)= \hat{x}(i\Delta T^*)$ for every $i=0,1,2,\cdots$, where
$$
\bar{A} = \begin{bmatrix}
A & - (B_p R_p^{-1} B_p'- B_e R_e^{-1}B_e')\tilde{K}\\
0 & A- (B_p R_p^{-1} B_p'- B_e R_e^{-1}B_e')\tilde{K}
\end{bmatrix}.
$$
At the discrete observation instances, the closed-loop system evolves according to
$$
\begin{bmatrix}
x\left((i+1)\Delta T^*\right)\\
\hat{x}\left((i+1)\Delta T^*\right)
\end{bmatrix}
 = e^{\bar{A}h}\begin{bmatrix}
x\left(i\Delta T^*\right)\\
\hat{x}\left(i\Delta T^*\right)
\end{bmatrix} + \begin{bmatrix}
\Id_n \\ \Id_n
\end{bmatrix} v_k,
$$
where $v_k = \int_{i\Delta T^*}^{(i+1)\Delta T^*} e^{A[(i+1)\Delta T^*-\tau]}Cdw(\tau)$. We know that if $e^{\bar{A}h}$ is Schur, we have $\mathbb{E}[x(i\Delta T^*)]= 0$ as $i\rightarrow \infty$ and $\sup_{i\geq 0}\mathbb{E}[\Vert x(i\Delta T^*)\Vert^2]< \infty$ \cite{kushner1971introduction}. To show $e^{\bar{A}h}$ is Schur, it is sufficient to show $\bar{A}$ is Hurwitz. Since $\hat{x}(i\Delta T^*) = x(i\Delta T^*)$ for every $i =1,2,\cdots$, we just need to show the system
$dz(t) = [A- (B_p R_p^{-1} B_p'- B_e R_e^{-1}B_e')\tilde{K}]z(t)dt$ 
is asymptotically stable  \cite{chen1999linear}. 

Consider a Lyapunov function $V(t) = z'\tilde{K}z$. Indeed,
$$
\begin{aligned}
\dot{V} &= \dot{z}'\tilde{K}z + z'\tilde{K}z\\
&\leq z(t)' [A' - \tilde{K}'(B_p R_p^{-1} B_p'- B_e R_e^{-1}B_e')']\tilde{K} z\\
&\ \ \ \ + z(t)' \tilde{K} [A- (B_p R_p^{-1} B_p'- B_e R_e^{-1}B_e')\tilde{K}]z(t)\\
&\leq z(t)' A'\tilde{K}z(t) - z(t)'[Q +  \tilde{K}A + A' \tilde{K}]z(t)  + z(t)'\tilde{K}Az(t)\\
&\ \ \ \ - z(t)'\tilde{K}(B_p R_p^{-1} B_p'- B_e R_e^{-1}B_e')\tilde{K}z(t)
\\
&=-z(t)'Qz(t) - z(t)'\tilde{K}(B_p R_p^{-1} B_p'- B_e R_e^{-1}B_e')\tilde{K}z(t) \leq 0,
\end{aligned}
$$
If  $B_p R_p^{-1} B_p'- B_e R_e^{-1}B_e'>0$, the Lyapunov stability theorem yields that $\tilde{K}z\rightarrow 0$ and $Q^{1/2}z\rightarrow 0$. Since $\tilde{K}z\rightarrow 0$, $z$ tends to the largest finite invariant set contained in $\{z:Q^{1/2} z =0\}$, for the system $dz = A z dt$. Suppose $(A, Q^{1/2})$ is observable, the largest finite invariant set is merely $x=0$. Hence, the system the system
$dz(t) = [A- (B_p R_p^{-1} B_p'- B_e R_e^{-1}B_e')\tilde{K}]z(t)dt$ 
is asymptotically stable.

Now, we can conclude that if $B_p R_p^{-1} B_p'- B_e R_e^{-1}B_e'>0$ and $(A,Q^{1/2})$ observable, then $\mathbb{E}[x(i\Delta T^*)]= 0$ as $i\rightarrow \infty$ and $\sup_{i\geq 0}\mathbb{E}[\Vert x(i\Delta T^*)\Vert^2]< \infty$. We know for $t\in[i \Delta T^*, (i+1)\Delta T^*)$, 
$$
\begin{aligned}
x(t) =& e^{A(t-i\Delta T^*)} x(i\Delta T^*)\\ 
&+ \int_{i\Delta T^*}^{t} e^{A(t-\tau)}(B_p R_p^{-1} B_p'- B_e R_e^{-1}B_e')\tilde{K}\hat{x}(\tau) d\tau\\
&+ \int_{i\Delta T^*}^{t} e^{A(t-\tau)}Cdw(\tau),
\end{aligned}
$$
where $\hat{x}(t)= e^{[A- (B_p R_p^{-1} B_p'- B_e R_e^{-1}B_e')\tilde{K}](t-i\Delta T^*)} x(i\Delta T)$. Hence, $\mathbb{E}[x(t)] = 0$ as $t\rightarrow \infty$. From \cref{FirstOrderPeriod}, we know $\Sigma(\Delta T^*)$ is bounded if the cost of observation $O_p$ is bounded. Then, if the cost of observation $O_p$ is bounded, $\sup_{t\geq 0}\mathbb{E}[\Vert x(t)\Vert^2]<\infty$.

\section{Numerical Experiments} \label{Sec: Experiements}





\begin{figure}
    \centering
    \includegraphics[width=1\columnwidth]{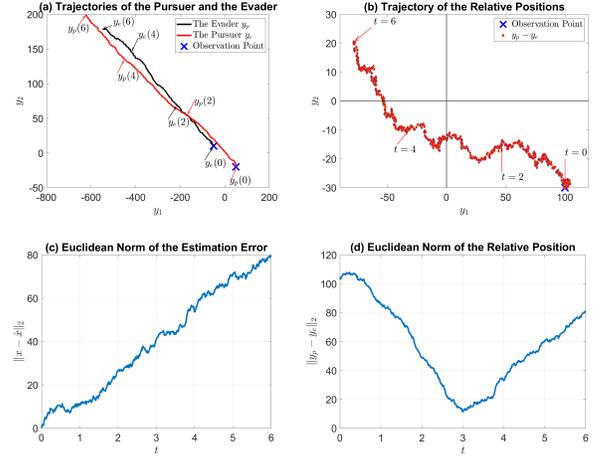}
    \caption{A realization of the PEEC game when $O_p = \infty$ and $c_p = c_e = \sqrt{16}$. (a) Trajectories of the Pursuer and the Evader on a two-dimensional plane; (b) Trajectory of the relative positions between the Pursuer and the Evader; (c) The Euclidean norm of the estimation error over time; (d) The Euclidean norm of the relative positions between the pursuer and the evader over time.}
    \label{fig:cp16T6N0}
\end{figure}

To illustrate the PEEC game and the Nash strategies, we consider a one pursuer and one evader game. The space is a planar surface for visualization purposes. Let $y_p\in\mathbb{R}^2$ be the 2-dimensional coordinates (position) of the pursuer, $z_p = \dot{y}_p\in\mathbb{R}^2$ is be velocity vector and $u_p$ be the acceleration control vector $\mathbb{R}^2$. Let $y_1$ and $y_2$ be the name of the two coordinates. Let $x_p = [y_p'\ z_p']'$ be the state of the pursuer, which includes the location and the velocity of the pursuer. The state of the pursuer is subject certain degree of disturbances which is captured by a $4$-dimensional standard Weiner process $w_p(t)\in\mathbb{R}^4$ for all $t$.  By physical law, the state dynamics of the pursuer is 
$$
dx_p(t) = A x_p(t) dt + B_pu_p(t) dt + C_pdw_p(t),
$$
where 
$$
A = \begin{bmatrix}
0 & 1 \\
0 & 0 \\
\end{bmatrix}\otimes\Id_2,\ \ \ B_p= \begin{bmatrix}
0 \\ 1
\end{bmatrix}\otimes \Id_2,\ \ \ C_p = c_p \cdot \Id_4.
$$
We define $y_e\in\mathbb{R}^2$ be the coordinates of the evader. Similarly, we have $z_e = \dot{y}_e$  and $x_e = [y_e'\ z_e']'$. The state dynamics of the evader can be described by $
dx_e(t) = A x_e(t) dt + B_eu_e(t) dt + C_edw_e(t),
$
where
$$
B_e= \begin{bmatrix}
0 \\ 1
\end{bmatrix}\otimes \Id_2,\ \ \ C_e = c_e \cdot \Id_4.
$$
Define a new state $x = x_p -x_e$. We have
$$
dx(t)  = Ax(t)dt  + B_p u_p(t)dt - B_e u_e(t) + Cdw(t),
$$
where $c = \sqrt{c_p^2 + c_e^2}$ and $(w(t),t \geq 0)$ is a $4$-dimensional standard Wiener process. The pursuer is trying to minimize the distance between him/her and the evader. The evader is trying to maximize it. Assume that acceleration on both axes require the same amount of effort/energy. Hence, we have
$$
Q = \begin{bmatrix}
1 & 0\\ 0 & 0\\ 
\end{bmatrix}\otimes \Id_2,\ \ \ R_p= \gamma R_e,\ \ \ R_e = 2\cdot \Id_2,
$$
where $\gamma \leq 1$. Let $Q_T = 10 \cdot Q$ and $\gamma = 0.8$. Let the terminal time $T=6s$. We set the initial positions and the initial velocities of the two players to be $x_p(0) = [50\ -20\  5\  10]'$ and $x_e(0) = [-50\ 10\ 1\ 10]'$. Parameters $c$, $\gamma$, $O_p$, and $O_e$ are subject to change. 

For numerical computation of the Nash observation strategies, we know that when $\gamma = 0.8$, the evader has less maneuverability than the pursuer. Hence, the evader's observation strategy is to not observe to expose himself/herself. To compute the pursuer's strategy, we first leverage the result given in \cref{ObservationNumberUpperBound} to compute the upper bound of the optimal number of observations $\bar{N}_p^*$. Then, for every $N_p \leq \bar{N}_p^*$, we solve the finite-dimensional optimization in \cref{innerObservationOptimization} using Algorithm \ref{algorithm1}.

In Figures \ref{fig:cp16T6N0}-\ref{fig:cp16T6N25}, we present the realizations of the PEEC game under various costs of observation when the system noise level is $c_p = c_e = \sqrt{16}$. In Figure \cref{fig:cp32T6N2}, we present a realization of the PEEC game when the optimal number of observations is $2$ and the system noise level is $c_p = c_e = \sqrt{32}$. To facilitate the visualization, we use animation to show the moving trajectories of the pursuer and the evader in the link \footnote{\url{https://github.com/Yun-Han/PE-DifferentialGame-StrategicInfo/tree/master/VideoSharing}}. We also add time indices $t=\{0,2,4,6\}$ to the figures to help readers visualize the moving trajectory.

\begin{figure}
    \centering
    \includegraphics[width=1\columnwidth]{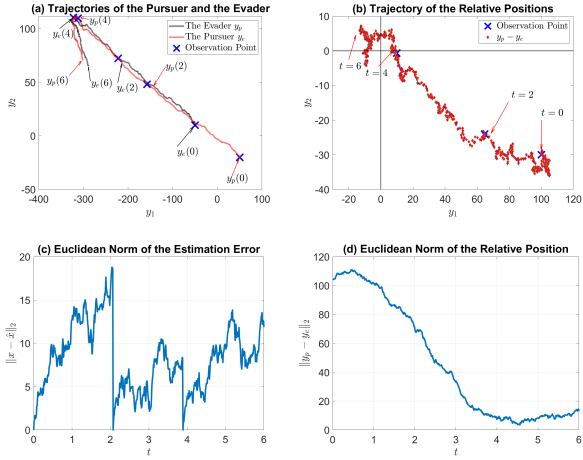}
    \caption{A realization of the PEEC game when $O_p = 900$ and $c_p = c_e = \sqrt{16}$. (a) Trajectories of the Pursuer and the Evader on a two-dimensional plane; (b) Trajectory of the relative positions between the Pursuer and the Evader; (c) The Euclidean norm of the estimation error over time; (d) The Euclidean norm of the relative positions between the pursuer and the evader over time.}
    \label{fig:cp16T6N2}
\end{figure}

When the cost of observation is infinity, i.e., $O_p=\infty$, the optimal observation strategy for the pursuer is to not observe at all. As we can see in fig. \ref{fig:cp16T6N0} (a), the only observation point (marked by a blue cross marker) is the initial conditions that are assumed to known to both players. In this case. the controls of both players are equivalent to the open-loop Nash control strategies in a deterministic setting. Since both players know each other's initial position, at the beginning, the evader escapes toward the exact opposite direction of where the pursuer is initially located. This is due to the fact that acceleration on $y_1$ axis and $y_2$ axis requires the same cost, i.e., $R_p$ and $R_e$ are identity matrices multiplied by some constants. As we can see from  in fig. \ref{fig:cp16T6N0} (d), the euclidean distance between the pursuer and the evader narrows. But as the estimation error accumulates due to no observation, the pursuer lose track of the evader and even goes beyond where the evader is actually located $y_e(6)$.

\begin{figure}
    \centering
    \includegraphics[width=1\columnwidth]{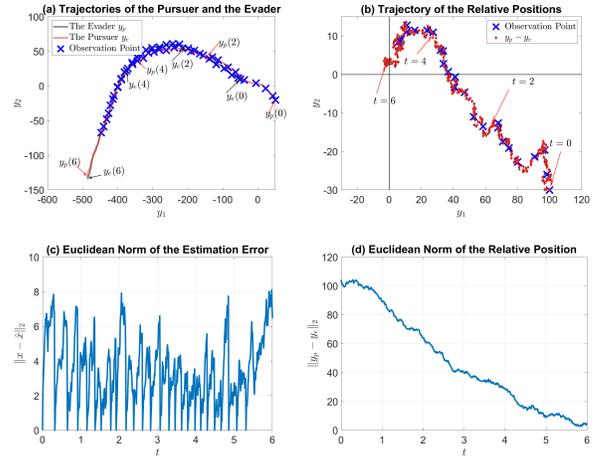}
    \caption{A realization of the PEEC game when $O_p = 10$ and $c_p = c_e = \sqrt{16}$. (a) Trajectories of the Pursuer and the Evader on a two-dimensional plane; (b) Trajectory of the relative positions between the Pursuer and the Evader; (c) The Euclidean norm of the estimation error over time; (d) The Euclidean norm of the relative positions between the pursuer and the evader over time.}
    \label{fig:cp16T6N25}
\end{figure}

When the cost of observation is $O_p = 900$, the optimal observation strategy for the pursuer is to observe two times at time instances $\mathcal{T}_p=\{2.06s, 3.87s\}$. Since when the pursuer observes, the evader also knows the pursuer's location at the same time. Hence, there are $6$ observation points for both players in Fig. \ref{fig:cp16T6N2} (a) including the initial points. Based on the initial condition, as in  \ref{fig:cp16T6N0} (a), the evader runs away from the pursuer and the pursuer chases after the evader following the same direction. At $t=2.06s$, the pursuer triggers the observation and both players observe each other's location. At this time, the relative position between the two players has the almost the same angle as the relative positives at time $0$, so the trajectory of the two players is almost a line until the next observation at $t=3.87s$. At $t=3.87s$, the pursuer and the evader receive each other's location and realize the relative angle between them is changed. Thus, after the observation, both players adjust their directions of chasing and evading, which cause a sharp turn in their trajectories. As we can see from Fig. \ref{fig:cp16T6N2} (c) that the estimate is refreshed to the actual state information and the estimation error is reset to $0$ when an observation arrives. From Fig. \ref{fig:cp16T6N2} (b), the relative position between the two players is close to the origin near the terminal time. And as is shown in \ref{fig:cp16T6N2} (d), the Euclidean distance of the relative position goes down to $5$ at the end, which is a relative low value compared with the Euclidean distance at the initial positions. This indicates that when the disturbances level $c_p = c_e =\sqrt{16}$, it is not necessary to observe every time to ensure a good performance. With an optimized set of observation time instances $\mathcal{T}_p$, the pursuer can also achieve a fairly good performance. Hence, the Nash observation strategy can also be used to help the pursuer save sensing/communication costs while maintaining a certain level of performance.

\begin{figure}
    \centering
    \includegraphics[width=1\columnwidth]{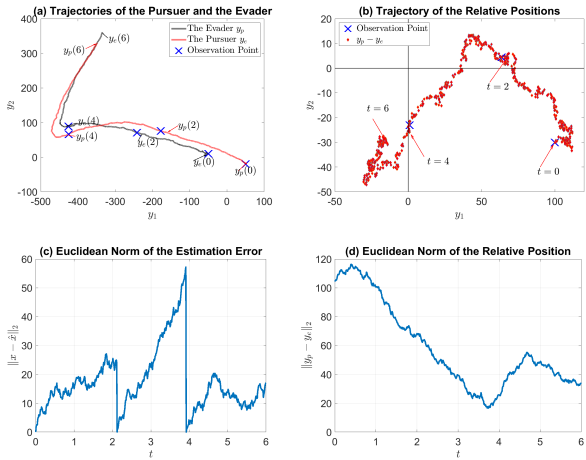}
    \caption{A realization of the PEEC game when $O_p = 10$ and $c_p = c_e = \sqrt{32}$. (a) Trajectories of the Pursuer and the Evader on a two-dimensional plane; (b) Trajectory of the relative positions between the Pursuer and the Evader; (c) The Euclidean norm of the estimation error over time; (d) The Euclidean norm of the relative positions between the pursuer and the evader over time.}
    \label{fig:cp32T6N2}
\end{figure}

If the cost of observation goes down to $O_p = 10$, it is optimal to observe $25$ times.
As we can see from Fig. \ref{fig:cp16T6N25} (a), the pursuer follows behind the evader and trajectories of two players overlap. We refer the readers to the animation provided in the link\footnote{\url{https://github.com/Yun-Han/PE-DifferentialGame-StrategicInfo/tree/master/VideoSharing}} for a clearer description of the trajectories. The pursuer senses frequently and as a result, the evader receives observation frequently. Hence, the pursuer and the evader adapts their controls immediately when they realize the angle of the relative position changes. The estimation error remains low as is shown in Fig. \ref{fig:cp16T6N25} (c). From Fig. \ref{fig:cp16T6N25} (b) and (d), we can see that with better maneuverability and frequent observations, the pursuer can easily narrows the distance to the evader to near zero before the terminal time. 

We increase the system disturbances level to $c_p = c_e = \sqrt{32}$. Fig. \ref{fig:cp32T6N2} presents a realization of the PEEC game when the optimal number of observations is $N_p^* =2$. Compared with the setting with lower disturbances, which is presented in Fig. \ref{fig:cp16T6N2}, the pursuer fails to narrow his/her distance to the evader to near zero when the system disturbances is larger. This shows that larger system disturbances give more advantage to an evader with less maneuverability when the pursuer has to pay a large overhead to sense. Hence, if an evader is less maneuverable than the pursuer, the evader can still escape if he/she can keep a high stealth level (makes it more expensive for the pursuer to observe). In military applications, this means stealth technologies are especially important for battlefield things with less maneuverability.

In conclusion, in this section, we show that a pursuer with higher maneuverability than the evader prefers more observations(exposures). But the pursuer can achieve reasonably good performance even when the number of observations is low. The Nash observation strategy enables the pursuer to observe less often while maintaining a good performance. We also show that when only a limited number of observations are available, larger system disturbances give an evader with less maneuverability more advantage. A less maneuverable evader can still escape if he/she can avoid being detected by his/her opponent frequently by making it more expensive for his/her opponent to observe.

\section{Conclusions}
This paper proposes a framework that introduces the concept of controlled information into PE differential games. This framework enriches the existing framework of PE differential games by capturing the interactions between the pursuer and the evader in the battlefield of information. We show that the Nash observation strategies depend only on the system characteristics. Players with less maneuverability won't observe at all in fear of the exposure of his/she own state. The proposed PEEC game has a symmetric information structure because when one player observes, the other player also obtains the information. With symmetric information structure, we avoid the second-guessing problem, which may render the problem untractable. The framework also sparks several exciting ideas for future exploring: 1. when one player senses(detects) the state(location) of the other player, he/she may expose his state (location), but the information received by the other player is noisier than what he/she receives. This scenario creates an asymmetric information game with noised observations. 2. future works can focus on analyzing the statistics aspects in terms of the players' performance, such as the probability of capture within a given time.

\appendix

\subsection{Proof of \Cref{Lemma:CostFunctionalSquare}}\label{Appendix:CostFunctionalSquare}
\begin{proof}
In this proof, we drop the time index of some variables for simplicity and readability purposes. The proof follows the arguments in the proof of Theorem {\rm II}.1 in \cite{duncan2015linear}.

Let $f(x,t) \coloneqq x(t)' K(t) x(t).$ An application of It\^{o}'s formula \cite{durrett2019probability} gives
{\small$$
\begin{aligned}
&df(x,t)\\
=& \frac{\partial f}{\partial t}(x,t) dt + \nabla_x f(x,t)' dx(t)+ \frac{1}{2}dx(t)' \nabla_{xx} f(x,t)dx(t) + o(dt),\\
=&x'\dot{K}x dt + x'K dx + dx'Kx + dx'Kdx + o(dt)\\
=& x' \left(\dot{K} + KA + A'K \right)x dt + x'K(B_p u_p - B_e u_e + Cw(t))dt\\
&+ (B_p u_p - B_e u_e+ Cw(t))'Kxdt + dw(t)'C'KCdw(t)dt  + o(dt),
\end{aligned}
$$}
where $\nabla_x$ and $\nabla_{xx}$ are the gradient and Hessian operators with respect to $x$ respectively. Immediately, we have
\begin{equation}\label{ItoIntegral}
\begin{aligned}
0 =& \mathbb{E}\Big[ \int_0^T x'(\dot{K} + KA + A'K)x + x'K(B_p u_p - B_eu_e) \\
&+ (B_pu_p -B_e u_e)'Kx dt \Big]\\
&+\int_0^T \Tr(KCC')dt -\left\{ \mathbb{E}\left[f\left(x(T),T\right) - f(x(0),0)  \right]\right\}.
\end{aligned}
\end{equation}
Adding the right-hand-side of \cref{ItoIntegral} to $J$ in \cref{CostFunctional} and completing the squares yield 
{\small
$$
\begin{aligned}
J =& \mathbb{E}\Big[ x(0)'K(0)x(0) - x(T)'K(T) x(T) + x(T)'Q_T x(T)\Big]\\
&+ \mathbb{E}\Big[\int_0^T x' \Big[ \dot{K} + KA + A'K+K\left(B_e R_e^{-1}B_e' - B_p R_p^{-1} B_p' \right) \Big]x' dt \Big]\\
&+ \mathbb{E}\Big[ \int_0^T\Vert u_p + R_p^{-1}B_p'Kx \Vert^2_{R_p} - \Vert  u_e +R_e^{-1} B'_eKx\Vert^2_{R_e}\\
&+ \Tr(KCC') dt + O_pN_p - O_eN_e \Big]\\
=&\Vert x_0 \Vert_{K(0)}^2 + \mathbb{E}\Big[ \int_0^T\Vert u_p + R_p^{-1}B_p'Kx \Vert^2_{R_p} - \Vert  u_e + R_e^{-1}B'_eKx\Vert^2_{R_e}\\
&+ \Tr(KCC') dt +O_pN_p - O_eN_e\Big].
\end{aligned}
$$}
This completes the proof.
\end{proof}

\subsection{Proof of \Cref{Theorem:NashControlStrategy}}\label{proof:NashControlStrategy}

\begin{proof}
In this proof, we drop the time index of some variables for simplicity and readability purposes. The proof follows follow a similar line of arguments as in \cite{bagchi1981linear,Maity2017ACC}. Given arbitrary $\Omega_p$ and $\Omega_e$, Player $\mathcal{P}$ aims to minimize $J$. Meanwhile, player $\mathcal{E}$ aims to maximize $J$. From \Cref{Lemma:CostFunctionalSquare}, we know that only the first two terms in \cref{CostFunctionalSquared} depend on the choices $u_p$ and $u_e$. Thus, the Nash control strategies can be obtained by solving the following problem
$$
\min_{u_p\in\mathcal{U}_p} \max_{u_e \in \mathcal{U}_e} J_c(u_p,u_e),
$$
where
$$
\begin{aligned}
J_c(u_p,u_e)\coloneqq& \mathbb{E} \bigg[ \int_{0}^T \Vert  u_p(t)+ R_p^{-1}B_p'K(t)x(t)\Vert^2_{R_p} \\
&- \Vert u_e(t)+  R_e^{-1}B_e'K(t)x(t)\Vert^2_{ R_e} dt\bigg].
\end{aligned}
$$
From Proposition 3.2 of \cite{Maity2017ACC}, we know that a necessary condition of a Nash control strategy is that $u_p$ lies in the range space of the linear operator $R^{-1}_p B_p'K$ and $u_e$ lies in the range space of the linear operator $R_e^{-1}B_e'K$. Since $\mathcal{U}_p$ and $\mathcal{U}_e$ are the sets of admissible control strategies that are progressively measurable with respect to $\mathcal{I}$. Thus, the Nash control strategies take the form of 
$$
\left(u_p(t),u_e(t)\right) = \left(R_p^{-1}B_p'K(t)\hat{x}_p(t),R_e^{-1}B_e'K(t)\hat{x}_e(t)\right),
$$
where $\hat{x}_p(t)$ and $\hat{x}_e(t)$, chosen by player $\mathcal{P}$ and player $\mathcal{E}$ respectively, have to be $\mathcal{I}(t)$ measurable.

The problem now becomes solving the following problem by finding $\hat{x}_p$ and $\hat{x}_e$ that are progressively $\mathcal{I}$ measurable:
$$
\begin{aligned}
\min_{\hat{x}_p} \max_{\hat{x}_e} \tilde{J}_c(\hat{x}_p,\hat{x}_e) \coloneqq & \int_0^T\mathbb{E}\Bigg[ \Vert x-\hat{x}_p \Vert_{KB_pR_p^{-1}B_p' K}^2 \\
&-  \Vert x - \hat{x}_e\Vert_{KB_e R_e^{-1}B_e' K}^2 \Bigg\vert \mathcal{I}(t) \Bigg]dt.
\end{aligned}
$$
Next, we study the first and second-order G\^ateaux differentials of $\tilde{J}_c$ to characterize a Nash strategy $(\hat{x}_p,\hat{x}_e)$. First, let's calculate the first and second-order of G\^{a}teaux differentials (pp.120 \cite{cheney2001analysis}) of $\tilde{J}_c$ at $(\hat{x}_p,\hat{x}_e)$ with directions $(h_e,h_e)$:
{\small
\begin{equation}\label{GateauxDerivative}
\begin{aligned}
&d_{(h_p,h_e)}\tilde{J}_c(\hat{x}_p,\hat{x}_e) \coloneqq \lim_{\epsilon\rightarrow 0} \frac{\tilde{J}_c(\hat{x}_p + \epsilon h_p,\hat{x}_e + \epsilon h_e) -\tilde{J}_c(\hat{x}_p,\hat{x}_e)}{\epsilon},\\
&d^2_{(h_p,h_e)}\tilde{J}_c(\hat{x}_p,\hat{x}_e)\coloneqq\\
&\lim_{\epsilon \rightarrow 0} \frac{\tilde{J}_c(\hat{x}_p + \epsilon h_p,\hat{x}_e +\epsilon h_e) - \tilde{J}_c(\hat{x}_p,\hat{x}_e) - \epsilon d_{(h_p,h_e)}\tilde{J}_c(\hat{x}_p,\hat{x}_e)}{\epsilon^2}.
\end{aligned}
\end{equation}}

Note that given $\hat{x}_p$ and $\hat{x}_e$, the solution of \cref{StochasticDynamics} can be expressed as
\begin{equation}\label{SolutionUnperturped}
\begin{aligned}
x(t) &= e^{At}x_0 - \int_{0}^t e^{A(t-s)} B_p R_p^{-1}B_p'K(s)\hat{x}_p(s) ds \\
&+ \int_0^t e^{A(t-s)}B_eR^{-1}_eB_e'K(s)\hat{x}_e(s)ds + \int_0^t e^{A(t-s)}C dw(s).
\end{aligned}
\end{equation}
Given the perturbations $\epsilon h_p$ and $\epsilon h_e$ on $\hat{x}_p$ and $\hat{x}_e$, the solution of \cref{StochasticDynamics} becomes
\begin{equation}\label{SolutionPerturped}
\tilde{x}(t) = x(t) - \epsilon H_p[h_p](t) + \epsilon H_e[h_e](t),
\end{equation}
where $H_p$ and $H_e$ are liear operators defined as
$$
\begin{aligned}
H_p[h_p](t) \coloneqq \int_0^t e^{A(t-s)}B_pR_p^{-1}B_pK(s)h_p(s)ds\\
H_e[h_e](t) \coloneqq \int_0^t e^{A(t-s)}B_eR_e^{-1}B_eK(s)h_e(s)ds.
\end{aligned}
$$
Therefore, we have
\begin{equation}\label{PerturbedCostFunctional}
\begin{aligned}
\tilde{J}_c(\hat{x}_p+\epsilon h_p, \hat{x}_e + \epsilon h_e) =  \int_0^T\mathbb{E}\Big[ \Vert \tilde{x}-\hat{x}_p \Vert_{KB_pR_p^{-1}B_p' K}\\ -  \Vert \tilde{x} - \hat{x}_e\Vert_{KB_e R_e^{-1}B_e' K}  \Big\vert \mathcal{I}(t) \Big]dt. 
\end{aligned}
\end{equation}
Using \cref{SolutionUnperturped,SolutionPerturped,PerturbedCostFunctional} in \cref{GateauxDerivative}, we have
\begin{equation}\label{FirstDerivative}
\small
\begin{aligned}
&\frac{1}{2}d_{(h_p,h_e)}\tilde{J}_c(\hat{x}_p,\hat{x}_e)=\\
&  \int_0^T \mathbb{E}\Bigg[ -\left[h_p(t) + H_p[h_p](t) - H_e[h_e](t)\right]'KB_pR_p^{-1} B_p' K(x-\hat{x}_p) \\
&+ \left[h_e(t) +H_p[h_p](t) - H_e[h_e](t)\right]'KB_e R_e^{-1} B_e'K(x-\hat{x}_e)\Bigg\vert \mathcal{I}(t)\Bigg]dt.
\end{aligned}
\end{equation}

The necessary condition for $(\hat{x}_p,\hat{x}_e)$ being a Nash strategy is $d_{(h_p,h_e)}\tilde{J}_c(\hat{x}_p,\hat{x}_e) =0$ for all possible directions $(h_p,h_e)$. Under this condition, both players have no incentives to move away from $(\hat{x}_p,\hat{x}_e)$. 

Here, we consider 
\begin{equation}\label{heConstruction}
h_e(t) = -\int_0^t e^{(A+B_eR_e^{-1} B_e'K)(t-s)} B_pR_p^{-1}B_p'K(s)h_p(s)ds.
\end{equation}
Hence, we have
{$$\small
\begin{aligned}
&H_e[h_e](t) \\
=& -\int_0^t e^{A(t-s)}B_eR_e^{-1}B_eK(s)\cdot\\ 
&\ \ \ \int_0^s e^{(A+B_eR_e^{-1} B_e'K)(s-\tau)} B_pR_p^{-1}B_p'K(\tau)h_p(\tau)d\tau ds\\
&=-\int_0^t \left[\int_{\tau}^t e^{A(t-s)}B_eR_e^{-1}B_eK(s) e^{(A+B_eR_e^{-1} B_e'K)(s-\tau)} ds\right]\cdot\\
&\ \ \  B_pR_p^{-1}B_p'K(\tau)h_p(\tau)d\tau \\
&=-\int_0^t \left[ \int_\tau^t \frac{d}{ds} e^{A(t-s)}e^{(A+B_eR_e^{-1}B_e'K)(s-\tau)}ds \right]\cdot\\
&\ \ \  B_pR_p^{-1}B_p'K(\tau)h_p(\tau)d\tau\\
&=-\int_0^t \left[  -e^{(A+B_eR_e^{-1}B_e'K)(t-\tau)} + e^{A(t-\tau)} \right]\cdot\\
&\ \ \ B_pR_p^{-1}B_p'K(\tau)h_p(\tau)d\tau\\
&= -H_p[h_p](t) -h_e(t).
\end{aligned}
$$}
That means for any $h_p$, we can construct $h_e$ following \cref{heConstruction} such that $h_e = H_e[h_e] - H_p[h_p]$. Hence, for all possible $h_p$, we have $h_e$ defined by \cref{heConstruction} such that
$$
\begin{aligned}
&\frac{1}{2}d_{(h_p,h_e)}\tilde{J}_c(\hat{x}_p,\hat{x}_e)\\
=& \int_0^T \mathbb{E}\left[-(h_p-h_e)'KB_p R_p^{-1}B_p'K(x-\hat{x}_p)\middle\vert \mathcal{I}(t) \right]dt.
\end{aligned}
$$
Hence, the necessary condition that makes sure $\frac{1}{2}d_{(h_p,h_e)}\tilde{J}_c(\hat{x}_p,\hat{x}_e)=0$ for all possible $h_p$ is 
$$
\mathbb{E}\left[ x(t) - \hat{x}_p(t) \middle \vert \mathcal{I}(t)\right],\ \ \ \textrm{for all }t.
$$
That means $\hat{x}_p(t) = \mathbb{E}\left[x(t) \middle\vert \mathcal{I}(t)\right]$. Similarly, for any $h_e$, we construct $h_p$ as
\begin{equation}\label{HpConstructed}
h_p(t) = \int_0^t e^{(A-B_pR_p^{-1} B_p'K)(t-s)} B_eR_e^{-1}B_e'K(s)h_e(s)ds,
\end{equation}
which gives $h_p = H_e[h_e] - H_p[h_p]$. For any given $h_e$ and $h_p$ constructed by \cref{HpConstructed}, we have
$$
\begin{aligned}
&\frac{1}{2}d_{(h_p,h_e)}\tilde{J}_c(\hat{x}_p,\hat{x}_e) \\
=& \int_0^T \mathbb{E}\left[(h_e-h_p)'KB_p R_p^{-1}B_p'K(x-\hat{x}_e)\middle\vert \mathcal{I}(t) \right]dt.
\end{aligned}
$$
Therefore, the necessary condition to guarantee that $d_{(h_p,h_e)}\tilde{J}_c(\hat{x}_p,\hat{x}_e)$ for all possible $h_e$ is
$$
\mathbb{E}\left[ x -\hat{x}_e\middle \vert \mathcal{I}(t)\right] =0,\ \ \ \textrm{for all } t.
$$
This implies $\hat{x}_p = \hat{x}_e = \mathbb{E}[x(t)\vert \mathcal{I}(t)]$. Note that $\mathcal{I}(t)= \{x(s)|0<s\leq t,s\in\mathcal{T}\}$, where $\mathcal{T}=\{t_1,t_2,\cdots,t_{N_p+N_e}\}$. Using the fact that $\mathbb{E}[\int_0^t e^{A(t-s)C}dw(s) \vert \mathcal{I}(t)]$ is a martingale \cite{durrett2019probability}, we obtain the following differential equation for $\hat{x}(t) \coloneqq\mathbb{E}[x(t)\vert \mathcal{I}(t)]$:
\begin{equation}\label{EstimatorDynamics}
\begin{aligned}
&d\hat{x}(t) = \left(A- (B_p R_p^{-1} B_p'- B_e R_e^{-1}B_e')K\right)\hat{x}(t) dt,\\
&\hat{x}(0) = x_0,\ \hat{x}(\tau) = x(\tau),\ \textrm{for all }\tau \in \mathcal{T}.
\end{aligned}
\end{equation}
To show the sufficiency of $(\hat{x}_p,\hat{x}_e) = (\hat{x},\hat{x})$ being a Nash equilibrium, we resort to the second order G\^{a}teaux differential defined in \cref{GateauxDerivative}. Following the definition in \cref{GateauxDerivative}, we calculate
$$
\begin{aligned}
&\frac{1}{2}d^2_{(h_p,h_e)}\tilde{J}_c(\hat{x}_p,\hat{x}_e)\\
=& \int_0^T \mathbb{E}\bigg[  \left\Vert h_p(t) + H_p[h_p](t) - H_e[h_e](t)  \right\Vert^2_{KB_p R_p^{-1} B_p'K} \\
&-  \left\Vert h_e(t) + H_p[h_p](t) - H_e[h_e](t)  \right\Vert^2_{KB_e R_e^{-1} B_e'K} \bigg\vert \mathcal{I}(t)\bigg].
\end{aligned}
$$
We need to show that at point $(\hat{x}_p,\hat{x}_e) = (\hat{x},\hat{x})$, there exist some directions $(h_p,h_e)$ such that $d^2_{(h_p,h_e)}\tilde{J}_c(\hat{x}_p,\hat{x}_e)<0$ and some other directions $d^2_{(h_p,h_e)}\tilde{J}_c(\hat{x}_p,\hat{x}_e) >0$. To show this, consider any $h_p\neq 0$ and  $h_e$ constructed according to \cref{heConstruction}. Then, let $h_p$ be a constant over time. We have $\frac{1}{2}d^2_{(h_p,h_e)}\tilde{J}_c(\hat{x}_p,\hat{x}_e) >0$. Similarly, we can show there exist some $(h_p,h_e)$ such that $\frac{1}{2}d^2_{(h_p,h_e)}\tilde{J}_c(\hat{x}_p,\hat{x}_e) <0$. This proves that $(\hat{x}_p,\hat{x}_e) = (\hat{x},\hat{x})$, where $\hat{x}$ has dynamics \cref{EstimatorDynamics}, constitutes a Nash control strategy of the PEEC game.
\end{proof}

\subsection{Proof of \Cref{Corollary:RewriteCostFunctional}}\label{Appendix:RewriteCostFunctional}

\begin{proof}
Using \cref{CostFunctionalSquared} in \cref{Lemma:CostFunctionalSquare} and the results in \Cref{Theorem:NashControlStrategy}, we know that 
$$
\begin{aligned}
\tilde{J}(\Omega_p, \Omega_e) &\coloneqq J(\Omega_p,u^*_p(\Omega_p,\Omega_e), \Omega_e,u^*_e(\Omega_p,\Omega_e))\\
&=\mathbb{E} \left[ \int_{0}^T \Vert  x-\hat{x}\Vert^2_{K(B_pR_p^{-1}B_p'-B_eR_e^{-1}B_e')K} dt\right]\\
&\ \ \ + O_pN_p - O_eN_e 
+\Vert x_0 \Vert^2_{K(0)}\\
&\ \ \ + \int_0^T  \Tr\left(K(t)CC'\right)dt.
\end{aligned}
$$
Using \cref{EstimatorDynamic0} and \cref{StochasticDynamics}, we know that 
$$
\begin{aligned}
dx - d\hat{x} &= \left(Ax - B_p R^{-1}_p B'_p K\hat{x} + B_e R^{-1}_e B'_e K\hat{x}\right)dt + C dw(t)\\
&\ \ \ \ -\left(A- (B_p R_p^{-1} B_p'- B_e R_e^{-1}B_e')K\right)\hat{x} dt,\\
&=A(x-\hat{x}) + Cdw(t)
\end{aligned}
$$
with refreshing points 
$$
\hat{x}(0) =x(0)= x_0,\ \hat{x}(\tau) = x(\tau),\ \textrm{for all }\tau \in \mathcal{T}.
$$
Thus, for any $t\in(0,T]$, let $\tau = \max\{s\ |\ s\in\mathcal{T},s<t\}$. We have
$$
\begin{aligned}
p(t)& \coloneqq x(t)-\hat{x}(t)\\
&= e^{A(t-\tau)} p(\tau)  + \int_\tau^t e^{A(t-s)} Cdw(s-\tau)\\
&= \int_\tau^t e^{A(t-s)}Cdw(s-\tau).
\end{aligned}
$$
Hence $\mathbb{E}[p(t)] = 0$. Let $P(t)\coloneqq \mathbb{E}[p(t)p(t)']$ be the variance of the estimation error. We have
\begin{equation}\label{EstimationVarianceDynamics}
P(t) = \int_\tau^t e^{A(t-s)} C C' e^{A(t-s)'} ds.
\end{equation}
Hence, we have
\begin{equation}\label{EstimationCostIntermediate}
\begin{aligned}
&\mathbb{E} \left[ \int_{0}^T \Vert  x-\hat{x}\Vert^2_{K(B_pR_p^{-1}B_p-B_eR_e^{-1}B_e)'K} dt\right] \\
=&\int_{0}^T\mathbb{E}\left[ p(t)'\left[K(B_pR_p^{-1}B_p'-B_eR_e^{-1}B_e')K\right]p(t) \right]dt\\
=& \int_{0}^T \Tr\left(P(t) K(B_pR_p^{-1}B_p-B_eR_e^{-1}B_e)'K\right)dt\\
=&\sum_{i=0}^{{N_p+N_e}} \int_{t_i}^{t_{i+1}} \Tr\Bigg[\left(\int_{t_i}^t e^{A(t-s)} CC'e^{A(t-s)}ds \right)\cdot\\
&\ \ \ \ K(B_pR_p^{-1}B_p'-B_eR_e^{-1}B_e')K \Bigg]dt\\
=&\sum_{i=0}^{{N_p+N_e}} \int_{t_i}^{t_{i+1}} \Tr\left[\Sigma(t-\tau_i) \varphi(t)\right]dt,\\
\end{aligned}
\end{equation}
where $\Sigma(t) = \int_0^t  e^{A(t-s)} C C' e^{A(t-s)'}dt$, $\varphi(t)= K(t)(B_pR_p^{-1}B_p'-B_eR_e^{-1}B_e')K(t)$, and $0=t_0 <t_1 \leq t_2 \cdots \leq t_{N_p+N_e}<t_{N_p+N_e+1} =T$. Hence, we complete the proof by showing 
$$
\begin{aligned}
\tilde{J}(\Omega_p, \Omega_e) &= \sum_{i=0}^{{N_p+N_e}} \int_{t_i}^{t_{i+1}} \Tr\left[\Sigma(t-\tau_i) \varphi(t)\right]dt + O_pN_p - O_eN_e \\
&\ \ \ 
+\Vert x_0 \Vert^2_{K(0)} + \int_0^T  \Tr\left(K(t)CC'\right)dt.
\end{aligned}
$$
\end{proof}

\subsection{Proof of \Cref{Proposition:ObservationStrategies}}\label{Appendix:ObservationStrategies}

\begin{proof}
First, we state two claims that are useful in the proof. 
\begin{claim}[Proposition 8.5.12 of \cite{bernstein2009matrix}]\label{Claim:TraceTwoMatrics}
Consider two symmetric matrices $\Sigma_1$ and $\Sigma_2$, and a positive semi-definite matrix $\Phi$. If $\Sigma_1 \leq \Sigma_2$, then $\Tr(\Sigma_1 \Phi) \leq \Tr(\Sigma_2 \Phi)$.
\end{claim}

\begin{claim}\label{Claim:VarianceTwoSets}
Let $(P_1,t\in[0,T])$ be the variance of the estimation error defined in \cref{EstimationVarianceDynamics} associated with $\mathcal{T}_1$, and $(P_2,t\in[0,T])$ be the variance of the estimation error defined in \cref{EstimationVarianceDynamics} associated with $\mathcal{T}_2$. If $\mathcal{T}_1\subset \mathcal{T}_2$, then $P_1(t)> P_2(t)$ for all $t\in [0,T]$.
\end{claim}

Here, \Cref{Claim:VarianceTwoSets} is a direct result of the definition of $(P(t),t\in[0,T])$ in \cref{EstimationVarianceDynamics}. To prove (\ref{Proposition:ZeroObservation}), let $\Omega_p = (N_p,\mathcal{T}_p)$ be any observation strategy of the pursuer. Let $\Omega_e^{no} = (0,\varnothing)$ be the no observation strategy for the evader. Let $\Omega_e$ be any other strategies such that $N_e \neq 0, \mathcal{T}_e\neq \varnothing$. Let $(P_1(t),t\in[0,T])$ be the variance of estimation error defined in \cref{EstimationVarianceDynamics} associated with $\mathcal{T}_1 = \mathcal{T}_p \cup \varnothing$ and let $(P_2(t),t\in[0,T])$ be associated with $\mathcal{T}_2 = \mathcal{T}_p \cup \mathcal{T}_e$. Hence, we have $\mathcal{T}_1 \subset \mathcal{T}_2$. By \Cref{Claim:VarianceTwoSets}, we have $P_1(t)\geq P_2(t)$ for all $t\in[0,T]$. From \cref{EstimationCostIntermediate}, we know
$$
\begin{aligned}
\tilde{J}_o(\Omega_p,\Omega_e^{no}) &= \int_0^T \Tr(P_1(t)\varphi(t))dt + N_pO_p\\
\tilde{J}_o(\Omega_p,\Omega_e) &= \int_0^T \Tr(P_2(t)\varphi(t))dt + N_pO_p - N_e O_e.\\
\end{aligned}
$$
By \Cref{Claim:TraceTwoMatrics} and the fact that $\varphi(t)$ is positive definite for all $t$ (this is true when $B_p R_p^{-1}B_p] > B_e R_e^{-1}B_e'$), we have $\tilde{J}_o(\Omega_p,\Omega_e^{no}) > \tilde{J}_o(\Omega_p,\Omega_e)$ for any $\Omega_p$ and any $\Omega_e\neq \Omega_e^{no}$. Thus, $\Omega_e^*=\Omega_e^{no}$.

Now we prove (\ref{Proposition:ObservationNumUpperBonund}). Since the optimal strategy for the evader is not to observe at all no matter what $\Omega_p$ is, the problem for the pursuer is to solve the following finite-dimensional optimization problem
$$
\min_{\Omega_p} \tilde{J}_o(\Omega_p,\Omega_e^{no})=\sum_{i=0}^{{N_p}} \int_{t_{p,i}}^{t_{p,i+1}} \Tr\left[\Sigma(t-t_{p,i}) \varphi(t)\right]dt + O_pN_p.
$$
When $O_p=0$, the best strategy is trivial, i.e., to observe every time and the optimal value will be $0$. When $O_e\neq 0$, suppose $\Omega_p^* = (N_p^*, \Omega_p^*)$ is the optimal strategy. We have
$$
\tilde{J}_o(\Omega_p^*,\Omega_e^{no}) \leq \tilde{J}_o((0,\varnothing),\Omega_e^{no})= \int_0^T \Tr\left[\Sigma(t)\varphi(t)\right]dt,
$$
and 
$$
\begin{aligned}
\tilde{J}_o(\Omega_p^*,\Omega_e^{no}) &= \sum_{i=0}^{{N^*_p}} \int_{t^*_{p,i}}^{t^*_{p,i+1}} \Tr\left[\Sigma(t-t_{p,i}) \varphi(t)\right]dt + O_pN^*_p\\
&\geq  O_pN^*_p.
\end{aligned}
$$
Combining the two inequalities above, we have \cref{ObservationNumberUpperBound}.

To prove (\ref{ObservationTimeFirstOrder}), note that for any given $N_p$, the optimal time instances $t^*_{p,i},i=1,2,\cdots,N_p$ has to satisfy the first-order necessary condition for the optimization problem given in \cref{innerObservationOptimization}. Taking derivatives on the objective function of \cref{innerObservationOptimization} with respect to $t_{p,i}$ and an application of Leibniz integral rule yield
{$$
\small
\begin{aligned}
&\frac{d}{dt_{p,i}} \sum_{j=0}^{{N_p}} \int_{t_{p,j}}^{t_{p,j+1}} \Tr\left[\Sigma(t-t_{p,j}) \varphi(t)\right]dt + O_pN_p\\
=&\frac{d}{dt_{p,i}}\Bigg\{ \int_{t_{p,i-1}}^{t_{p,i}} \Tr\left[\Sigma(t-t_{p,i-1}) \varphi(t)\right]dt + O_pN_p \\
&+ \int_{t_{p,i}}^{t_{p,i+1}} \Tr\left[\Sigma(t-t_{p,i}) \varphi(t)\right]dt + O_pN_p\Bigg\}\\
=&\Tr\left[\Sigma(t_{p,i}-t_{p,i-1})\varphi(t_{p,i})\right] +\int_{t_{p,i}}^{t_{p,i+1}} \Tr\left[\frac{d}{dt_{p,i}}\Sigma(t-t_{p,i})\varphi(t) \right]dt,\\
=&\Tr\left[\Sigma(t_{p,i}-t_{p,i-1})\varphi(t_{p,i})\right]\\
&\ \ \ - \int_{t_{p,i}}^{t_{p,i+1}} \Tr\left[e^{A(t-t_{p,i})} CC' {e^{A(t-t_{p,i})}}'\varphi(t) \right]dt,\\
=&\int_{t_{p,i-1}}^{t_{p,i}}\Tr\left[ e^{A(t_{p,i}-t)}CC'{e^{A(t_{p,i}-t)}}'\varphi(t_{p,i})\right]dt\\
&\ \ \ -\int_{t_{p,i}}^{t_{p,i+1}} \Tr\left[e^{A(t-t_{p,i})} CC' {e^{A(t-t_{p,i})}}'\varphi(t) \right]dt
\end{aligned}
$$}
where we used the fact that
$$
\begin{aligned}
\frac{d}{dt_{p,i}}\Sigma(t-t_{p,i}) &= \frac{d}{dt_{p,i}}\int_{t_{p,i}}^t  e^{A(t-s)} C C' e^{A(t-s)'}ds\\
&=-e^{A(t-t_{p,i})} CC' {e^{A(t-t_{p,i})}}'.
\end{aligned}
$$
Since the objective function in \cref{innerObservationOptimization} in continuous in $t_{p,i}$ for every $i=1,2,\cdots,N_p$ and the constraint set is a closed and bounded subset of $\mathbb{R}^{N_p}$ (hence compact), by Weierstrass extreme value theorem, there exists at least one minimizer for the optimization problem in \ref{innerObservationOptimization}. Thus, we arrive the conclusions in (\ref{ObservationTimeFirstOrder}).

\end{proof}


%

\ifCLASSOPTIONcaptionsoff
  \newpage
\fi



\bibliographystyle{IEEEtran}
\bibliography{ControlSharedEconomy.bib}
\end{document}